\documentclass{article}
\usepackage{amsthm}
\usepackage{amsfonts}
\usepackage{fullpage} 
\usepackage{tikz-cd}
\usepackage{quiver}
\usepackage{amsmath}
\usepackage{amssymb}
\usepackage{bbm}
\usepackage{color}
\usepackage{MnSymbol} 
\usepackage{parskip}
\usepackage[section]{placeins} 
\usepackage{url}
\usepackage{comment}
\usepackage{trimclip}
\usepackage{hyperref}
\definecolor{navy}{RGB}{0,0,128}
\hypersetup{
    colorlinks=true,
    linkcolor=navy,
    urlcolor=navy,
    citecolor=navy
}

\theoremstyle{plain}
\newtheorem{theorem}{Theorem}[section]
\newtheorem{thm}[theorem]{Theorem}

\newtheorem{proposition}[theorem]{Proposition}

\newtheorem{corollary}[theorem]{Corollary}

\theoremstyle{definition}
\newtheorem{definition}[theorem]{Definition}
\newtheorem{example}[theorem]{Example}

\newtheorem{remark}[theorem]{Remark}

\newcommand{\dfn}[1]{\emph{#1}}
\renewcommand{\dfn}[1]{\textbf{#1}}

\newcommand{\twobb}{\mathbbm{2}}
\newcommand{\onebb}{\mathbbm{1}}
\newcommand{\ibb}{{\ensuremath{\mathbb{ I}}} }

\newcommand{\id}{\mathrm{id}}
\newcommand{\Id}{\mathrm{Id}}
\newcommand{\dom}{\mathrm{dom}}
\newcommand{\cod}{\mathrm{cod}}

\newcommand{\co}{\mathrm{co}}
\newcommand{\op}{\mathrm{op}}
\newcommand{\Set}{\mathsf{Set}}
\newcommand{\Pos}{\mathsf{Pos}}
\newcommand{\DL}{\mathsf{DL}}

\newcommand{\BA}{\mathsf{BA}}
\newcommand{\Pri}{\mathsf{Pri}}

\newcommand{\dProx}{{\mathbb P}\mathsf{rox}}
\newcommand{\PriPre}{\mathsf{PriPre}}
\newcommand{\dPriPre}{{\mathbb P}\mathsf{riPre}}

\newcommand{\Prox}{\mathsf{Prox}}
\newcommand{\ProxRel}{\mathsf{ProxRel}}
\newcommand{\SubDL}{\mathsf{SubDL}}
\newcommand{\SubBA}{\mathsf{SubBA}}
\newcommand{\DGF}{\mathsf{DGF}}

\newcommand{\KOHaus}{\mathsf{KOHaus}}

\DeclareMathOperator{\Spec}{Spec}

\newcommand{\then}{\mathrel{;}}
\newcommand{\dRel}{{\mathbb R}\mathsf{el}}

\newcommand{\Rel}{\mathsf{Rel}}

\newcommand{\Map}{\mathsf{Map}}
\newcommand{\Spl}{{\mathbb S}\mathsf{pl}}

\newcommand{\Filt}{\mathsf{Filt}}
\newcommand{\RFilt}{\mathsf{RFilt}}
\newcommand{\RPFilt}{\mathsf{RPFilt}}

\newcommand{\ddown}{%
\mathord{\mathchoice
  {\raise2pt\hbox{%
  \ooalign{\hss$\downarrow$\hss\cr\lower2pt\hbox{%
  $\downarrow$}}}}
  {\raise2pt\hbox{%
  \ooalign{\hss$\downarrow$\hss\cr\lower2pt\hbox{%
  $\downarrow$}}}}
  {\raise1.5pt\hbox{%
  \ooalign{\hss$\scriptstyle\downarrow$\hss\cr\lower1.5pt\hbox{%
  $\scriptstyle\downarrow$}}}}
  {\raise1.1pt\hbox{%
\ooalign{\hss$\scriptscriptstyle\downarrow$\hss\cr\lower1.1pt\hbox{%
  $\scriptscriptstyle\downarrow$}}}}
}}
\newcommand{\uup}{%
\mathord{\mathchoice
  {\raise2pt\hbox{%
  \ooalign{\hss$\uparrow$\hss\cr\lower2pt\hbox{%
  $\uparrow$}}}}
  {\raise2pt\hbox{%
  \ooalign{\hss$\uparrow$\hss\cr\lower2pt\hbox{%
  $\uparrow$}}}}
  {\raise1.5pt\hbox{%
  \ooalign{\hss$\scriptstyle\uparrow$\hss\cr\lower1.5pt\hbox{%
  $\scriptstyle\uparrow$}}}}
  {\raise1.1pt\hbox{%
\ooalign{\hss$\scriptscriptstyle\uparrow$\hss\cr\lower1.1pt\hbox{%
  $\scriptscriptstyle\uparrow$}}}}
}}

\makeatletter
\newcommand{\twoheaddownarrowCustom}{\mathrel{\mathpalette\twoheaddownarrow@\relax}}
\newcommand{\twoheaddownarrow@}[2]{%
  \begingroup
  \sbox\z@{$\m@th#1\downarrow$}%
  \vphantom{\copy\z@}%
  \ooalign{\copy\z@\cr\hidewidth\clipdownarrow@{#1}\hidewidth\cr}%
  \endgroup
}
\newcommand{\clipdownarrow@}[1]{%
  \raisebox{0.25\ht\z@}{\clipbox{-3pt 0pt -3pt {0.5\height}}{\copy\z@}}%
}
\newcommand{\twoheaduparrowCustom}{\mathrel{\mathpalette\twoheaduparrow@\relax}}
\newcommand{\twoheaduparrow@}[2]{%
  \begingroup
  \sbox\z@{$\m@th#1\uparrow$}%
  \vphantom{\copy\z@}%
  \ooalign{\copy\z@\cr\hidewidth\clipuparrow@{#1}\hidewidth\cr}%
  \endgroup
}
\newcommand{\clipuparrow@}[1]{%
  \raisebox{0.1\ht\z@}{\clipbox{-3pt {0.65\height} -3pt 0pt}{\copy\z@}}%
}
\makeatother
\newcommand{\dda}{\mathord{\twoheaddownarrowCustom}}
\newcommand{\uua}{\mathord{\twoheaduparrowCustom}}

\title{From Zero-Dimensional to Continuous Dualities: A Double-Categorical Account}

\author{Alexander Kurz, M. Andrew Moshier, Achim Jung}
\date{\today}

\begin{document}

\maketitle

\begin{abstract}
We investigate how to systematically construct continuous dualities from zero-dimensional dualities, employing well-known methods from algebra, topology, category theory, and domain theory. 
While our method is general, this paper focusses on the move from Stone spaces to compact Hausdorff spaces and the move from Priestley spaces to compact ordered Hausdorff spaces. 
The engine of our approach is Stone duality for relations: on the space side quotienting by a preorder turns zero-dimensional spaces into continuous ones, while distributive lattices with a proximity relation are their algebraic duals.
Our duality for relations is inherently order-enriched. Double categories organise both functional and relational morphism in the same structure. 
The move from zero-dimensional to continuous dualities is then a three-step construction: extend a duality from functional to relational morphism, split idempotents, restrict to maps.
\end{abstract}

\setcounter{tocdepth}{1}
\tableofcontents

\newpage
\section{Introduction}\label{sec:introduction}

We investigate how to systematically construct continuous dualities from zero-dimensional dualities, employing well-known methods from algebra, topology, category theory, and domain theory. 
To avoid overloading the paper with long preliminaries on these subjects, we look at the special case of the move from Stone spaces to compact Hausdorff spaces, or, rather, its order-enriched generalisation, the move from Priestley spaces to compact ordered spaces (also called compact pospaces or Nachbin spaces). 

Classical theorems connecting algebraic and continuous structures include Manes' theorem (every compact Hausdorff space is the quotient of a Stone space) \cite[III.2.4]{JohnstoneStoneSpaces}, Gleason's theorem characterising the minimal projective covers of compact Hausdorff spaces \cite[III.3.7]{JohnstoneStoneSpaces}, and the theorem that every continuous lattice is a quotient of an algebraic lattice \cite[Thm~I-4.17]{GierzEtAl}. 
An order-enriched generalisation of Manes' theorem is due to Flagg \cite{Flagg}, see also Goubault-Larrecq's blog post \cite{GoubaultLarrecq2022FilterMonads}.

From a technical point of view, our work derives from the Jung-S\"underhauf duality for stably compact spaces \cite{JungSunderhauf1996} and the Jung-Kegelmann-Moshier duality for proximity lattices \cite{JungKegelmannMoshier1999,Kegelmann1999}. 
One of the main insights of those works is to base the equivalence between algebraic and continuous structures on (1) equipping the algebraic structures with a relation (standing in for way-below \cite[VII.2.2]{JohnstoneStoneSpaces}) and (2) weakening isomorphism of algebraic structures to isomorphism with respect to relational morphisms.

In order to bring these results into a unified framework, we add the following ingredients to classical Stone-type dualities.

\begin{itemize}
  \item We work in the context of order-enriched categories. For example, as is the case already in \cite{JohnstoneStoneSpaces}, the duality between Stone spaces and Boolean algebras is understood as the restriction of the duality between Priestley spaces and distributive lattices. 
  But we make crucial use of the fact that these (and similar) dualities involve equivalences between order-enriched categories (making the order between the arrows part of the categorical structure). 
  \item We extend the classical Stone-type dualities to dualities for relations as carried out in some detail in \cite{KMJ}. This allows us to account for dualities that involve algebraic structures  containing relations (such as proximity lattices and subordination algebras) and to account for relational isomorphism.
  \item The move from zero-dimensional to continuous structures can then be explained via the following general category-theoretic  constructions:
  \begin{enumerate}
    \item Extend a zero-dimensional duality to relations.
    \item Split idempotents in the category of relations.
    \item Restrict to adjoint pairs (maps).
  \end{enumerate}
  These steps form the middle layer of the synopsis in the next section.
  \item In order to formalise the constructions outlined in the previous item at the right level of abstraction, we rely on the formalism of double categories which allows us to have functional morphisms and relational morphisms in the same mathematical structure.
\end{itemize}

We assume that the reader is familiar with Stone and Priestley duality
\cite{DaveyPriestley2002,GehrkeVanGool2024} as well as basic category theory
\cite{MacLane1998,KellyEnriched,Grandis2019}.
We will explain some of the key ideas of our construction using the example of the Cantor space in \emph{Section~\ref{sec:cantor}} and review in \emph{Section~\ref{sec:preliminaries}}  preliminaries on order-enriched and double categories.

\emph{Section~\ref{sec:dualities-relations}} presents an interlude showing  how to use double categorical Stone duality to recover and extend known results on subordination algebras. This is based on having subordinations as relational arrows and algebra morphisms as functional arrows in the same double category. 
\emph{Section~\ref{sec:splitting-idempotents}} starts on the path towards continuous dualities. It splits idempotents in the double categories of Priestley relations and of distributive lattice relations. This construction preserves duality.
\emph{Section~\ref{sec:proxlat}} spells out the algebra side as proximity lattices and reconstructs the (strong) proximity lattice of a compact ordered space. 
\emph{Section~\ref{sec:priestley-preorders}} spells out the space side as Priestley preorders and shows that every Priestley preorder quotients to a compact ordered space. 
\emph{Section~\ref{sec:dedekind-cuts}} characterises points of proximity lattices as adjoint pairs of relations (cuts) and, in the strong case, as round prime filters. 
\emph{Section~\ref{sec:round-filters}} studies round filters and roundification, and shows that for a strong proximity lattice the quotient of $\Spec(A)$ by the dual preorder is a compact ordered space (isomorphic to the space of round prime filters). 
\emph{Section~\ref{sec:adjoint-relations-in-priestley-preorders}} shows that adjoint relations between Priestley preorders correspond to continuous monotone maps between compact ordered spaces. 
\emph{Section~\ref{sec:adjoint-relations-proximity-lattices}} dualises to adjoint relations between proximity lattices. 
\emph{Section~\ref{sec:conclusion}} summarises and points to further work.

\paragraph{Notation} We use $X$, $x$, etc.\ for entities on the space side of a duality and $A$, $a$ etc.\ for entities on the algebra side.
When referring to an object of a general category (or double category), we use $C$, $D$, etc.
The two-element dualising object is denoted by $\twobb$ on either side. For example, $\twobb^X$ denotes the set (or distributive lattice) of continuous order-preserving maps from a Priestley space $X$ to the Priestley space $\twobb$ while $\twobb^A$ denotes the set (or Priestley space) of homomorphisms from a distributive lattice $A$ to the distributive lattice $\twobb$. A relation $R$ on spaces will be written $xRy$ and its dual $a\twobb^R b$. Similarly, $aRb$ for a relation on algebras and $x\twobb^R y$ for its dual.

We write $\mathcal{C}$ for general categories, $\mathbb{C}$ for general double categories, and $\bf C$ for the Cantor space. 

We use $\id$ for identity functional arrows and $\Id$ for identity relational arrows.
We write $C^o$ for the order-dual of a poset (thin category) $C$ and $\mathcal C^{\mathrm{op}}$ for the opposite of a category $\mathcal C$.

By definition of $\twobb^A$ for distributive lattices $A$, we may identify a point $p\in\twobb^A$ with its prime filter $p^{-1}(1)$. The natural order on points is the element-wise order on $\twobb^A$, so the order is inclusion of prime filters. The algebraic (ring-theoretic) tradition of Davey and Priestley favours identifying a point with its kernel $p^{-1}(0)$. In ordinary categories, this is just a choice of convention: is $1$ drawn at the top or at the bottom? It matters only when order-enrichment comes into play. Most of the time in this paper, the reader may safely ignore this convention. We will draw attention to it only when needed in Section~\ref{sec:round-filters}.

\section{Synopsis}\label{sec:synopsis}

At the bottom, we have the duality between Priestley spaces and distributive lattices. 
Going up, we embed into their double categories of relations. 
Going up again, we complete by splitting idempotents, on spaces by splitting idempotents above identity (Priestley preorders) and on algebras by splitting idempotents below identity (proximity lattices). 
Next, the dashed arrows again embed categories of maps into categories of relations, so that going up corresponds to restricting to maps.
The top row reconstructs the dual equivalence between compact ordered spaces and proximity lattices. Note that each of the three steps (category of relations, splitting idempotents, restricting to maps) is equivalence-preserving.

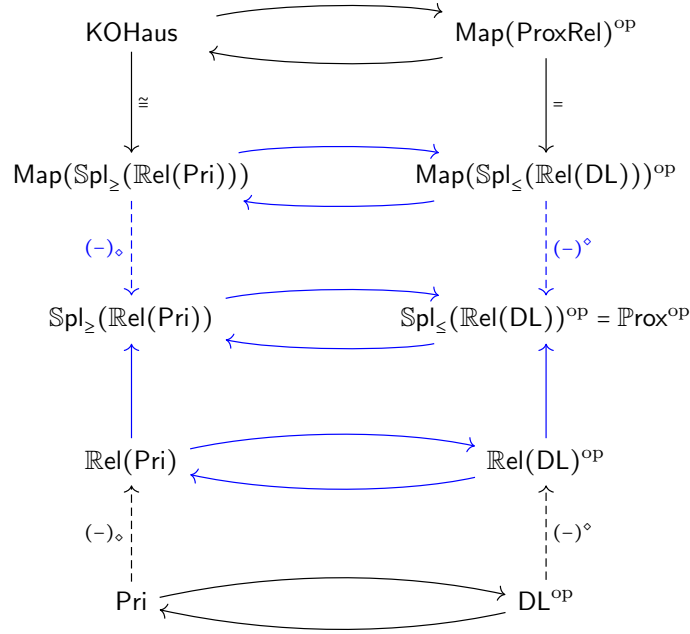
\[\begin{tikzcd}
	{\quad\KOHaus\quad} && {{\Map(\ProxRel)}^{\rm op}} \\
	\\
	{\sf \Map(\Spl_\ge(\dRel(\Pri)))} && {\sf \Map(\Spl_\le(\dRel(\DL)))^{\rm op}} \\
	\\
	{\Spl_\ge(\dRel(\Pri))} && {\Spl_\le(\dRel(\DL))^{\rm op} = \dProx^\op} \\
	\\
	{\sf \dRel(Pri)} && {\sf \dRel(DL)^{\rm op}} \\
	\\
	{\sf Pri} && {{\sf DL}^{\rm op}}
	\arrow[curve={height=-12pt}, from=1-1, to=1-3]
	\arrow["\cong", from=1-1, to=3-1]
	\arrow[curve={height=-12pt}, from=1-3, to=1-1]
	\arrow["=", from=1-3, to=3-3]
	\arrow[blue, curve={height=-12pt}, from=3-1, to=3-3]
	\arrow[blue, dashed, "{(-)_\diamond}"', from=3-1, to=5-1]
	\arrow[blue, curve={height=-12pt}, from=3-3, to=3-1]
	\arrow[blue, dashed, "{(-)^\diamond}", from=3-3, to=5-3]
	\arrow[blue, curve={height=-12pt}, from=5-1, to=5-3]
	\arrow[blue, curve={height=-12pt}, from=5-3, to=5-1]
	\arrow[blue, from=7-1, to=5-1]
	\arrow[blue, curve={height=-12pt}, from=7-1, to=7-3]
	\arrow[blue, from=7-3, to=5-3]
	\arrow[blue, curve={height=-12pt}, from=7-3, to=7-1]
	\arrow[dashed, "{(-)_\diamond}", from=9-1, to=7-1]
	\arrow[curve={height=-12pt}, from=9-1, to=9-3]
	\arrow[dashed, "{(-)^\diamond}"', from=9-3, to=7-3]
	\arrow[curve={height=-12pt}, from=9-3, to=9-1]
\end{tikzcd}\]%

\section{The unit interval as a quotient of Cantor space}\label{sec:cantor}

This is the motivating example that is worth keeping in mind. Let ${\bf C}$ be the Cantor space. ${\bf C}$ can be represented as the set of infinite sequences $\{0,1\}^\omega$, or as the middle third subset of the unit interval $\ibb$.

Let $\equiv$ be the equivalence relation identifying the two endpoints of each middle third gap. 
Quotienting ${\bf C}$ by ${\equiv}$ gives (a space homeomorphic to) the unit interval $\ibb$.
For example $0111 \ldots$ and $1000\ldots$ are identified.

It will be important that $\equiv$ is idempotent, that is $({\equiv}\mathrel{;}{\equiv}) = {\equiv}$, and above identity, that is, $\Id_{\bf C}\subseteq{\equiv}$.

The dual ${\bf A}$ of ${\bf C}$ is the set of clopens of ${\bf C}$. ${\bf A}$ can be described as generated by finite sequences $\{0,1\}^\ast$ in the following way. Each $w\in\{0,1\}^\ast$ is identified with the clopen $\underline w$ consisting of all infinite words starting with $w$. ${\bf A}$ is then the closure of all $\big\{\underline w \mid w\in\{0,1\}^\ast\big\}$ under finite unions and intersections.

The dual $\prec$ of $\equiv$ is the relation ${\prec} : {\bf A} \looparrowright {\bf A}$ given by $a\prec b$ if the closure of $a$ under $\equiv$ is contained in $b$. Thus $a\prec b$ if $a$ is contained in the $\mathbb{I}$-interior of $b$, that is, if $a$ is \dfn{way-below} $b$. Intuitively, $a\prec b$ if there is some room on the left and right of $a$ to find $c$ such that $a\prec c\prec b$. Indeed, we can always find clopens $a_l$ and $a_r$ such that $a\prec a_l\cup a\cup a_r\prec b$.

Formally, as we will discuss in more detail in the next section, the dual of a relation $R\subseteq X\times Y$ is given, for $a\subseteq X$ and $b\subseteq Y$, by $a \twobb^R b$ if $R[a] \subseteq b$. In particular, the dual of the identity relation on ${\bf C}$ is subset inclusion of clopens of ${\bf A}$. 

If we order elements of ${\bf A}$ (clopens) by inclusion $\subseteq$, then ${\prec}$ is below the identity of ${\bf A}$ which we take to be the ordering of clopens by inclusion ${\le_{\bf A}} = {\subseteq}$. In fact, ordering clopens by $\prec$, the only reflexive elements are bottom and top (empty set and ${\bf C}$).
$\prec$ is idempotent and ${\prec}\subseteq {\Id_{\bf A}}$. As noted above, this yields interpolants between related clopens.

The Boolean algebra ${\bf A}$ of clopens is the Stone dual of the Cantor space (equivalently, the free Boolean algebra on countably many generators). The proximity relation $\prec$ is dual to the equivalence relation $\equiv$. 
As we will see later in detail, the category of so-called proximity lattices $(A,\prec)$ is dual to the category of zero-dimensional spaces equipped with suitable preorders. For the special case of equivalence relations see Remark~\ref{rem:ba-prox-equiv}.

Let us briefly indicate here how that works in this example. It is well-known from Stone/Priestley duality that we can obtain ${\bf C}$ from ${\bf A}$ by ${\bf C}=\DL({\bf A},\twobb)$ and ${\bf A}$ from ${\bf C}$ by ${\bf A}=\Pri({\bf C},\twobb)$. We now show that $({\bf A},\prec)$ and $({\bf C},\equiv)$ are dual in a similar way.

The dual of $({\bf A},\prec)$ is given by all `cuts' of $({\bf A},\prec)$ which, as we will show, are `maps' $({\bf A},\prec)\to\twobb$. A cut, in turn, is a roundification%
\footnote{A filter $F\subseteq {\bf A}$ is round if $a\in F \ \Rightarrow \ \exists b\in F\,.\, b\prec a$.} 
of a prime filter. 
For example, a clopen $[x,1/3]$ can only be in a roundification if there is also a smaller clopen not containing the endpoint 1/3. 
Therefore, in order to approximate 1/3 one needs to use intervals $[x,1/3]\cup[2/3,y]$. It follows that roundifications of prime filters cannot distinguish 1/3 and 2/3, which, therefore, are glued together in the dual.

Conversely, the dual of $({\bf C},\equiv)$ is again the cuts of ${\bf C}$, which in this case correspond to clopens.

As we will explain in the next section, the right setting for duality of relations is order-enriched, so it takes no extra effort to develop our theory for Priestley spaces and distributive lattices instead of Stone spaces and Boolean algebras. But the relevant topological concepts are already visible in the order-discrete examples. 

\section{Preliminaries on order-enriched and double categories}\label{sec:preliminaries}

This section collects the preliminaries needed for the paper. We do not aim for a comprehensive treatment of order-enriched or double categories, but rather explain why these notions arise naturally in the study of dualities for relations and introduce only the terminology helpful for this paper.

\subsection{Why order-enriched categories?}
\label{sec:why-order-enriched}

The main reason we work in the order-enriched setting is the following: in an order-enriched category, every (weakening) relation can be both \dfn{tabulated} as a span and \dfn{co-tabulated} as a cospan. In the set-enriched setting this fails: if one tabulates a relation as a span and then takes the pullback of the pushout, the resulting relation does not in general coincide with the original; similarly, if one sends a $\Set$-relation to a $\BA$-relation and back. 

In the order-enriched setting, using comma and cocomma objects in place of pullbacks and pushouts, tabulation and co-tabulation are mutually inverse in the appropriate sense, and the duality of distributive lattices and Priestley spaces extends to a duality of their categories of relations \cite{KMJ}.

To see how this works, let us denote by $\twobb^-$ the dualising functor given by homming into the two-element chain $\twobb$ (considered, depending on context, as a Boolean algebra, Stone space, Priestley space, distributive lattice). Let $R \subseteq X \times Y$ be a relation and let $p:R\to X$ and $q:R\to Y$ be the projections. Then the dual of $R$ is the relation $\twobb^R$ given for all $a\in \twobb^X$, $b\in\twobb^Y$ by

\begin{equation}\label{eq:dual-relation}
a \twobb^R b \ \Leftrightarrow \ \twobb^p(a)\subseteq\twobb^q(b) \ \Leftrightarrow \ R[a] \subseteq b.
\end{equation}

The first equivalence arises from dualising the tabulating span $X\xleftarrow{p} R\xrightarrow{q} Y$ into the cospan $\twobb^X\xrightarrow{\twobb^p} \twobb^R\xleftarrow{\twobb^q} \twobb^Y$ and taking the comma object, while the second equivalence is familiar from weakest precondition semantics \cite{DijkstraScholten90,BonsangueKok1994,Hasuo2014,Jacobs2014} and modal logic, reading $R[a] \subseteq b$ as `all $R$-successors of $a$ are in $b$'.

Note that even if $X,Y,R$ are discrete, the dual relation $\twobb^R$ is ordered by inclusion due to $\twobb^p(a)\subseteq\twobb^q(b)$. In other words, even if one was only interested in the duality between Boolean relations and Stone relations, the duality passes by distributive lattices and Priestley spaces. 

More generally, in our setting, a \dfn{relation} $R : X \looparrowright Y$ is a subobject of the product that is \dfn{weakening} in the sense that

\[ x' \le_X x \ , \ x \, R \, y\ , \ y \le_Y y' \quad\Longrightarrow\quad \ x' \, R \, y'. \] 

Equivalently, the characteristic function of the relation $X^\op \times Y \to \twobb$ is monotone. Note that weakening relations are closed under composition and that the identity relation on $X$ is given by $\le_X$.

Since the notion of subobject depends on the category in question we also speak of Priestley relations (or $\Pri$-relations) and distributive lattice relations (or $\DL$-relations). In particular, in addition to being weakening-closed,

\begin{itemize}
\item a $\Pri$-relation is a topologically closed subspace;
\item a $\DL$-relation $R:A\looparrowright A'$ satisfies
  \begin{gather*}
  0 \mathrel R 0 \quad\quad 1 \mathrel R 1 \\[0.7em]
  \frac{a \mathrel R b \quad\quad c\mathrel R d}{a \vee c \, \mathrel R \, b \vee d} \quad\quad \frac{a   \mathrel R b \quad\quad c\mathrel R d}{a \wedge c \, \mathrel R \, b \wedge d}
   \end{gather*}
\end{itemize}
In particular, a $\DL$-equivalence relation is a congruence relation. 

In \cite{KMJ}, we investigate the theory of weakening relations in a general order-enriched setting. In this paper, we focus on particular order-enriched categories and do not need the general theory. It therefore suffices for us to say the following.

An \dfn{order-enriched category} is a category where each hom-set carries a partial order preserved by composition. Functors between order-enriched categories are locally monotone, that is, they preserve the order on hom-sets. 

For our investigation, the central example of an order-enriched category is that of compact ordered spaces.

\begin{definition}\label{def:compact-ordered-space}
A \dfn{compact ordered space} \cite[I.3]{NachbinTopologyAndOrder} is a compact topological space $X$ equipped with a partial order that is closed as a subset of $X\times X$.
\end{definition}

\begin{remark}\label{rem:nachbin-separation}
Compact ordered spaces are Hausdorff, since they are normally ordered \cite[I.2]{NachbinTopologyAndOrder}, that is, they satisfy \dfn{Nachbin separation}
\begin{equation}\label{eq:nachbin-separation}
\begin{gathered}
D \text{ closed down-set},\;
U \text{ closed up-set},\;
D\cap U=\emptyset
\\[0.3em]
\Longrightarrow\quad
\exists\, d\text{ lower open},
\exists\, u\text{ upper open}\text{ such that } \;
D\subseteq d,\;
U\subseteq u,\;
d\cap u=\emptyset.
\end{gathered}
\end{equation}
Synonyms for compact ordered space include compact order-Hausdorff space \cite[VII.2.1]{JohnstoneStoneSpaces}, compact pospace \cite[9.1]{Goubault-LarrecqNonHausdorff}, and Nachbin space.
\end{remark}

The most important examples of order-enriched categories for us are:
\begin{itemize}
  \item $\KOHaus$: Compact ordered spaces with continuous monotone maps (Definition~\ref{def:compact-ordered-space}).
  \item $\Pri$: Priestley spaces with continuous monotone maps.%
  \footnote{A Priestley space $X$ is a compact ordered space satisfying Priestley separation: for every closed downset $D$ and disjoint closed upset $U$ there is a clopen downset $C$ with $D\subseteq C$ and $(X\setminus C)\supseteq U$.}. 
  Priestley separation is stronger than Nachbin separation.
  \item $\DL$: Distributive lattices with homomorphisms.
  \item $\Rel(\Pri)$: Priestley spaces with closed weakening relations.
  \item $\Rel(\DL)$: Distributive lattices with $\DL$-relations.
\end{itemize}

In all of these categories, the order between functional arrows is pointwise and the order between relational arrows is by inclusion. The identity relation on an object $(C,\le_C)$ is given by ${\le_C}$. It is co-tabulated by the cospan $(\id_C,\id_C)$ and tabulated by the comma object of $(\id_C,\id_C)$.

\begin{remark}[Dual properties of relations]\label{rem:dual-properties-relations}
A useful consequence of the order on relations is that it allows us to express how some important properties dualise: duality reverses the order between relations, so we have the following dual properties of relations $R: C\looparrowright C$:
\begin{itemize}
  \item \dfn{above identity} (${\le_C} \subseteq R$, reflexive) $\longleftrightarrow$ \dfn{below identity} ($R \subseteq {\le_C}$).
  \item \dfn{transitive} ($R\,;R \subseteq R$) $\longleftrightarrow$ \dfn{interpolative} ($R \subseteq R\,; R$).
  \item \dfn{monad} (idempotent above identity, closure operator) $\longleftrightarrow$ \dfn{comonad} (idempotent below identity, interior operator).
\end{itemize}
Hence a Priestley preorder (reflexive, transitive) corresponds dually to a proximity relation (interpolative below identity). Equivalently, preorders are idempotent relations above identity and proximity relations are idempotent relations below identity.
\end{remark}

The correspondences above depend only on the order-theoretic structure of the categories of relations. If one restricts to relations on order-discrete structures, one can also talk about symmetric relations:\footnote{A poset $X$ is discrete if and only if for all weakening relations $R:X\looparrowright X$ the converse of $R$ is weakening.}

\begin{remark}[Symmetry]\label{rem:ba-prox-equiv}
Recall that the dual $R$ of a closed relation $S$ on a Stone space is a $\DL$-relation on a Boolean algebra.
\begin{itemize}
  \item $S$ is \dfn{symmetric} iff $aRb \Leftrightarrow\neg b R\neg a$, 
\end{itemize}
see also (S6) in \cite[Lem.~4.6]{BBSV17}.
Thus a Boolean proximity dualises to an equivalence relation precisely when this self-duality holds. The Cantor space proximity in Section~\ref{sec:cantor} is of that kind.
\end{remark}

\subsection{A duality theorem for relations}

We start by observing that order-enriched Priestley duality

\[\begin{tikzcd}
	{\sf Pri} && {{\sf DL}^{\rm op}}
	\arrow["\twobb^-", curve={height=-12pt}, from=1-1, to=1-3]
	\arrow["\twobb^-", curve={height=-12pt}, from=1-3, to=1-1]
\end{tikzcd}\]

does (of course) reverse arrows, but does not reverse the order because composition (hence homming into $\twobb$) is order preserving.

We have seen above that Priestley duality when applied to a relation $R:X\looparrowright Y$ results in the dual relation $\twobb^R:\twobb^X\looparrowright \twobb^Y$ given by

$$a \twobb^R b \ \Leftrightarrow \ \twobb^p(a)\subseteq\twobb^q(b) \ \Leftrightarrow \ R[a] \subseteq b.$$

Our aim in this section is to formulate this result as an equivalence of order-enriched categories. To this end, we need to first introduce two functors mapping functions to relations. 

Given a concretely order-enriched category $\mathcal C$, for every arrow $f:C\to D$ there are relations
\[
  f_\diamond:C\looparrowright D,\quad (c,d)\mapsto D(fc,d),\qquad
  f^\diamond:D\looparrowright C,\quad (d,c)\mapsto D(d,fc).
\]

Here we follow category-theoretic conventions and identify relations $C\looparrowright D$ with maps 

$$C^o\times D\to \twobb,$$ 

order-reversing in the first argument and order-preserving in the second, so that we can identify $D(fc,d)$ with $f(c)\le_D d$ and $D(d,fc)$ with $d\le_D f(c)$. Sometimes $f_\diamond=\{(c,d)\in C\times D \mid f(c)\le d\}$ is called the hypergraph \cite{JungSunderhauf1996} while $f^\diamond=\{(c,d)\in C\times D \mid d\le_D f(c)\}$, the hypograph, appears for example  in Scott's theory of approximable mappings \cite{Scott1982}.
In the terminology of double categories we call $f_\diamond$ the \dfn{companion} of $f$ \cite{GP04,Grandis2019}, and $f^\diamond$ the \dfn{conjoint} of $f$ \cite{Shulman}.

As is well-known in the categorical theory of distributors (aka profunctors or bimodules), of which the order-enriched weakening relations are a special case, it follows from the definitions of $f_\diamond, f^\diamond$ that they are adjoint:

$$f_\diamond \dashv f^\diamond.$$ 

It is not hard to verify that $(-)_\diamond$ and $(-)^\diamond$ are functorial. In fact, $(-)_\diamond$ is covariant on arrows and contravariant on the order, whereas $(-)^\diamond$ is contravariant on arrows and covariant on the order. More formally,

$$(-)_\diamond: \Pri\to\Rel(\Pri)^\co
\quad\quad\quad\quad
(-)^\diamond: \DL^\op\to\Rel(\DL)$$

where the label $^\co$ refers to the category in which the order between the relations is reversed.

\begin{remark}\label{rem:companions-conjoints}
Abstractly, a companion (resp.\ conjoint) of an arrow in a double category is a proarrow equipped with unit and counit squares satisfying the triangle identities \cite{GP04,Shulman}.
A double category in which every arrow has both is an \dfn{equipment}, see eg \cite{Wood1982,CarboniKellyVerityWood1998,Shulman,Grandis2019,Koudenburg2019,Lambert2022,Nasu2025}.
Extending the current work to this axiomatic framework is ongoing and future work.
\end{remark}
We can now state Theorem 5.9 of \cite{KMJ}. 

\begin{theorem}[Duality for Relations]\label{thm:duality-for-relations}
The dual equivalence between $\Pri$ and $\DL$ extends to a dual equivalence between the order-enriched categories of $\Pri$-relations and $\DL$-relations
\[\begin{tikzcd}
	{\sf \Rel(Pri)^{\co}} && {\sf \Rel(DL)^{}} \\
	\\
	{\sf Pri} && {{\sf DL}^{\rm op}}
	\arrow["\twobb^-", curve={height=-12pt}, from=1-1, to=1-3]
	\arrow["\twobb^-", curve={height=-12pt}, from=1-3, to=1-1]
	\arrow["{(-)_\diamond}", from=3-1, to=1-1]
	\arrow["\twobb^-", curve={height=-12pt}, from=3-1, to=3-3]
	\arrow["{(-)^\diamond}"', from=3-3, to=1-3]
	\arrow["\twobb^-", curve={height=-12pt}, from=3-3, to=3-1]
\end{tikzcd}\]
\end{theorem}

Note that this duality does not reverse the direction of the relations, only the order between relations. This aligns well with viewing relations as additional structure best represented by objects, rather than as morphisms. As we will see, in double categories, relational morphisms are objects (in the category of arrows).

\subsection{Why double categories?}
\label{sec:why-double-categories}

The main reason to turn to double categories is that even if one is interested only in relational morphisms, it often pays to also keep the functional morphisms around in the same structure. Double categories are structures in which relational morphisms and functional morphisms can coexist as so-called `vertical morphisms' and `horizontal morphisms'.

In \cite{KMJ}, the motivation to generalise Theorem~\ref{thm:duality-for-relations} from order-enriched categories to double categories was to generalise the extension theorem from equivalences to adjunctions. In Section~\ref{sec:dualities-relations} we will see another application of double categories to relations.

Before we can state the double category-theoretic duality theorem for relations we need the following definitions. Recall from Section~\ref{sec:why-order-enriched} that for the order-enriched categories $\mathcal C$ considered in this paper, we can define a notion of relation, now called \dfn{$\mathcal C$-relation}. We show how to assemble $\mathcal C$-relations and $\mathcal C$-arrows into the structure of a double category. For background on order-enriched categories we refer to Kelly \cite{KellyEnriched}; for double categories to Grandis \cite{Grandis2019}.

\begin{definition}\label{def:double-category}
A (strict) \dfn{double category} $\mathbb C$ is a category internal to the category of categories
\[\begin{tikzcd}
\mathbb C_1 \arrow[r, shift left, "\mathrm{dom}"] \arrow[r, shift right, "\mathrm{cod}" swap] & \mathbb C_0
\end{tikzcd}\]
\end{definition}
\begin{remark}
  We will avoid the terminology of `horizontal' and `vertical' arrows as conventions differ between authors \cite{GrandisPare,Shulman}. While we focus on applications of double categories to categories with relations arising from concrete order-enriched categories, it is sometimes convenient for us to use a more abstract language that can be used for double categories in general. We settled on the compromise recorded in Table~\ref{tab:doublecat-conventions}. Identity proarrows are denoted by $\Id_C:C\looparrowright C$ and identity arrows by $\id_C:C\to C$. For arrows $f,g:C\to D$, we write $f\le g$ if $(f,g):\Id_C\to \Id_D$ is a square. For proarrows $R,R': C \looparrowright D$, we write $R\le R'$ if $(\id_C,\id_D):R\looparrowright R'$ is a square.
\end{remark}

\begin{table}[h!]
  \centering
  \renewcommand{\arraystretch}{1.4}
  \begin{tabular}{|l|l|l|}
    \hline
    & \textbf{Concrete Category} & \textbf{Abstract Category} \\
    \hline
    Objects of $\mathbb{C}_0$ & object & object \\
    \hline
    Arrows of $\mathbb{C}_0$ & (functional) morphism & arrow \\
    \hline
    Objects of $\mathbb{C}_1$ & relation, relational morphism & proarrow \\
    \hline
    Arrows of $\mathbb{C}_1$ & $c\,R\,d\Rightarrow f(c)\,S\,g(d)$ & square  \\
    \hline
  \end{tabular}
  \caption{Conventions regarding double categories.}
  \label{tab:doublecat-conventions}
\end{table}

Since we will apply the construction of a double category of relations to a number of different categories, it is convenient to have a general definition. The notion of a concretely order-enriched category which we use here has been discussed in detail in \cite{KMJ}. We do not repeat it here as we hope that it will be clear how to construct a double category of relations in each instance we need it (mainly for $\mathcal C=\Pri$ or $\DL$).

\begin{definition}\label{def:drel}
For any concretely order-enriched category $\mathcal C$, its \dfn{double category of relations} $\dRel(\mathcal C)$ is the double category given as follows:
\begin{itemize}
\item $\mathbb C_0$ is the category $\mathcal C$.
\item $\mathbb C_1$ is the category that has as objects (proarrows) $\mathcal C$-relations $R:C\looparrowright D$ and as arrows (squares) pairs of maps $(f:C\to C',g:D\to D')$ in $\mathcal C$
\[
\begin{tikzcd}[column sep=1.2em, row sep=1.2em]
  C \ar[r, "f"] \ar[d, "R"'] & C' \ar[d, "R'"] \\
  D \ar[r, "g"'] & D'
\end{tikzcd}
\]
such that $cRd \ \Rightarrow \ f(c)R'g(d)$.
\item A square $(f,g):R\to R'$ with $R:C\looparrowright D$ and $R':C'\looparrowright D'$ is mapped by $\dom:\mathbb C_1\to\mathbb C_0$ to $f:C\to C'$ and by $\cod:\mathbb C_1\to\mathbb C_0$ to $g:D\to D'$.
\end{itemize}
\end{definition}

\begin{definition}\label{def:arrow-proarrow-proj}
A double category of relations $\mathbb C$ has an \dfn{arrow projection}, obtained by restricting to squares whose proarrows are identities, and a \dfn{proarrow projection}, obtained by restricting to squares whose arrows are identities.
\end{definition}

In our examples, arrow and proarrow projections are order-enriched. A flat double category has pre-order enriched projections and an ultraflat double category has poset-enriched projections:

\begin{definition}\label{def:flat-ultraflat}
A double category is \dfn{flat} if every double cell is determined by its boundary \cite[3.2.1]{Grandis2019}. It is \dfn{ultraflat} if it is flat and invertible squares with vertical or horizontal identities are identity squares \cite[3.2.4]{Grandis2019}. 
\end{definition}

\begin{example}
$\dRel(\Pri)$ and $\dRel(\DL)$ have arrow projections $\Pri$ and $\DL$, respectively, and proarrow projections $\Rel(\Pri)$ and $\Rel(\DL)$, respectively. 
More generally, $\dRel(\mathcal C)$ is strict and ultraflat: relations compose on the nose, squares are a property of the boundary, and the projections are poset-enriched.
\end{example}

\begin{remark}
This point of view also allows us to view order-enriched categories such as $\Pri$ and $\Rel(\Pri)$ as particular cases of double categories.
\end{remark}

\begin{definition}\label{def:opposite-double-cat}
The \dfn{opposite} $\mathbb C^\op$ of a double category $\mathbb C$ is obtained by reversing arrows and squares, while leaving proarrows unchanged.
\end{definition}

\begin{remark}
Given a double category $\mathbb C_1\rightrightarrows\mathbb C_0$, relational morphisms are objects of $\mathbb C_1$ and functional morphisms are arrows of $\mathbb C_0$. Reversal of arrows in $\mathbb C_0$ forces reversal of squares in $\mathbb C_1$ because $\dom,\cod:\mathbb C_1\to\mathbb C_0$ are functors. In particular, comparison cells $(\id_C,\id_D):R\to S$ reverse, so the order on proarrows is reversed, while comparison cells $(f,g):\Id_C\to\Id_D$ do not reverse.
\end{remark}

The following duality theorem for the double category of distributive lattice relations is an instance of \cite[Thm 7.11]{KMJ} and generalises Theorem~\ref{thm:duality-for-relations}.

\begin{theorem} \label{thm:duality-double-rel} The dual equivalence between $\Pri$ and $\DL$ extends to a dual equivalence of double categories between the double category of $\Pri$-relations and the double category of $\DL$-relations
\[\begin{tikzcd}
	{\sf \dRel(Pri)^{}} && {\sf \dRel(DL)^{\op}}
	\arrow["\twobb^-", curve={height=-12pt}, from=1-1, to=1-3]
	\arrow["\twobb^-", curve={height=-12pt}, from=1-3, to=1-1]
\end{tikzcd}\]
where $^\op$ is the opposite double category of Definition~\ref{def:opposite-double-cat}.
\end{theorem}

\begin{remark}\label{rem:op-arrows-proarrows}
To see in more detail how the dualising functor $\twobb^-$ interacts with the $^\op$ of the theorem, consider the following square and its dual:
\[
  \begin{tikzcd}[column sep=2.2em, row sep=2em]
    A \ar[r, "f"] \ar[d, "R"'] & A' \ar[d, "S"] &
    \makebox[0em]{} &
    \twobb^A \ar[d, "\twobb^R"'] & \twobb^{A'} \ar[l, "\twobb^f"', above] \ar[d, "\twobb^S"] \\
    B \ar[r, "g"'] & B' &
    \makebox[0em]{} &
    \twobb^B & \twobb^{B'} \ar[l, "\twobb^g"]
  \end{tikzcd}
  \]
Taking $R,S$ to be identities, we see that $f\le g \ \Leftrightarrow \ \twobb^f \le \twobb^g$ and taking $f,g$ to be identities, we see that $R\le S \ \Leftrightarrow \ \twobb^S \le \twobb^R$. 
This aligns with what we said earlier: duality reverses the order between relations but not the direction of relations, and reverses the direction of functions but not the order between functions. 
\end{remark}

\begin{remark}
  It is interesting to see how condition \eqref{eq:dual-relation} of Section~\ref{sec:why-order-enriched} arises from the more abstract double categorical point of view. Start with the square in $\dRel(\Pri)$ on the left-hand side below, where $a,b$ are clopen subsets of $X,Y$, respectively.
  \[
  \begin{tikzcd}[column sep=2.2em, row sep=2em]
    X \ar[r, "a"] \ar[d, "R"'] & \twobb \ar[d, "\Id_\twobb"] &
    \makebox[0em]{} &
    \twobb^X \ar[d, "\twobb^R"'] & \onebb \ar[l, "a"', above] \ar[d, "\Id_\onebb"] \\
    Y \ar[r, "b"'] & \twobb &
    \makebox[0em]{} &
    \twobb^Y & \onebb \ar[l, "b"]
  \end{tikzcd}
  \]
  The squares being dual means that the two square conditions must be equivalent, that is, 
  $$
  xRy \Rightarrow a(x)\le b(y) \ \ \Longleftrightarrow \ \ a \twobb^R b
  $$
  which is equivalent to $R[a] \subseteq b$ as in \eqref{eq:dual-relation}.
  \end{remark}

\section{Dualities from double categories of relations}\label{sec:dualities-relations}

The aim of this section is to show that the double categorical duality of relations can be used to unify and extend a number of classical results from the theory of subordination algebras.

We restrict Theorem~\ref{thm:duality-double-rel} to the subcategory of $\dRel(\DL)$ where the objects are certain endo-relations (proarrows) $R:A\looparrowright A$ and the morphisms are certain squares, namely arrows $f:A\to A'$ such that $aRb \ \Rightarrow \ f(a)R'f(b)$. Such $\DL$-relations are often called \dfn{subordinations} and a distributive lattice $A$ with a subordination is called a \dfn{subordination algebra}.

\begin{definition}\label{def:subordination}
A \dfn{subordination} $\prec$ on a distributive lattice (or a Boolean algebra) $A$ is a $\DL$-relation $A\looparrowright A$. The (order-enriched) category $\SubDL$ of \dfn{subordination algebras} is the subcategory of the double category of $\DL$-relations with subordination algebras $(A,\prec)$ as objects and arrows $f:A\to A'$ with $a\prec b\ \Rightarrow f(a) \prec' f(b)$. The category of \dfn{$\BA$-based subordinations} $\SubBA$ further restricts to Boolean algebras $A$.
\end{definition}

We now go through a list of corollaries to the duality theorem for $\DL$-relations. 

\begin{corollary}\label{cor:subdl-pri-rel}
The category of Priestley spaces with closed weakening relations is dually equivalent to $\SubDL$.
\end{corollary}

The next corollary is a variation of \cite[Thm 9]{Celani-2001} and \cite[Thm 3]{DV-2006}, see \cite[Remark 2.23]{BBSV17}.

\begin{corollary}\label{cor:subba-stone}
The category $\SubBA$ is dually equivalent to the category of Stone spaces with closed relations.
\end{corollary}

The next corollaries exploit Remark~\ref{rem:dual-properties-relations}: for example, \emph{reflexive} is dual to \emph{below identity},\footnote{Recall that the identity relation on a Boolean algebra or a distributive lattice is their natural order. In particular, the dual of the identity relation (the diagonal) on a Stone space is the natural order on the dual Boolean algebra.} \emph{transitive} is dual to \emph{interpolative}.
All of these results are restrictions of the corresponding result on Priestley spaces and  $\DL$-based subordinations, which we do not state explicitly.

\begin{corollary}\label{cor:reflexive-below-id}
The category of Stone spaces with reflexive closed relations is dually equivalent to the category of $\BA$-based subordinations below identity.
\end{corollary}

\begin{corollary}\label{cor:transitive-interpolative}
The category of Stone spaces with transitive closed relations is dually equivalent to the category of $\BA$-based interpolative subordinations.
\end{corollary}
 
\begin{corollary}\label{cor:preorder-interpolative-below-id}
The category of Stone spaces with reflexive and transitive closed relations (preordered Stone spaces) is dually equivalent to the category of $\BA$-based interpolative subordinations below identity.
\end{corollary}

Combining with Remark~\ref{rem:ba-prox-equiv} yields the equivalence case of Section~\ref{sec:cantor} on Cantor space.

\begin{corollary}\label{cor:equivalence-self-dual}
The category of Stone spaces with closed equivalence relations is dually equivalent to the category of $\BA$-based interpolative subordinations below identity satisfying $a\prec b\Leftrightarrow\neg b\prec\neg a$.
\end{corollary}

Corollaries~\ref{cor:transitive-interpolative}, \ref{cor:preorder-interpolative-below-id}, and~\ref{cor:equivalence-self-dual} match \cite[Thm.~4.11(1)--(3)]{BBSV17}.

\begin{definition}\label{def:dgf}
A \dfn{descriptive general frame} ($\DGF$) $(X,R)$ is a Stone space with a closed relation $R$ such that $R^{-1}$ preserves clopens. 
The inverse image of $R$ preserves clopens iff, for any clopen $b$, the set $\{a\mid a\prec b\}$ has a maximal element (given by $\Box_R b$). Such a relation $\prec$ is called \dfn{modally definable}, see \cite[Def 2.9, Lem 4.3]{BBSV17}.
\end{definition}

We now obtain \cite[Thm 4.5(1)]{BBSV17} as another corollary.

\begin{corollary}\label{cor:modally-definable-dgf}
The category of $\BA$-based modally definable subordination algebras is dually equivalent to the category of $\DGF$s with relational morphisms.
\end{corollary}

The constraints stemming from the order of relations and the constraint of preservation of clopens/modal definability can be combined in a modular fashion to obtain further results, for example:

\begin{corollary}\label{cor:reflexive-dgf}
The category of reflexive $\DGF$s with relational morphisms is dually equivalent to the category of modally definable $\BA$-based subordinations below identity.
\end{corollary}

\begin{corollary}\label{cor:transitive-dgf}
The category of transitive $\DGF$s with relational morphisms  is dually equivalent to the category of modally definable $\BA$-based interpolative subordinations.
\end{corollary}

\begin{corollary}\label{cor:preorder-dgf}
The category of reflexive and transitive $\DGF$s with relational morphisms is dually equivalent to the category of modally definable $\BA$-based interpolative subordinations below identity.
\end{corollary}

\begin{remark}[Morphisms vs J\'onsson--Tarski]\label{rem:jonsson-tarski-morphisms}
The dualities above are weaker on morphisms than classical J\'onsson--Tarski duality: subordination morphisms need not preserve $\Box$ on the nose (they only force $f\Box\le\Box'f$), and dually one allows relational morphisms between $\DGF$s rather than continuous p-morphisms. Restricting on both sides---to Boolean algebra morphisms with $f(\Box a)=\Box'f(a)$, and to continuous p-morphisms---yields the classical duality: under Stone duality, $f\Box=\Box'f$ is equivalent to the dual map being a p-morphism (lax $f\Box\le\Box'f$ matches ``forth''; colax $\Box'f\le f\Box$ matches ``back'').
\end{remark}

We phrased the results in terms of the more widely considered $\BA$-based subordinations but all of our proofs are obtained by restricting the corresponding results about ($\DL$-based) subordination algebras. 

\section{Splitting idempotents}\label{sec:splitting-idempotents}

Returning to the development outlined in the synopsis (Section~\ref{sec:synopsis}), in this section we will explain in detail the step from 

$$\dRel(\Pri) \simeq \dRel(\DL)^\op$$

to

$$\Spl_\ge(\dRel(\Pri)) \simeq \Spl_\le(\dRel(\DL))^\op$$

obtained by splitting idempotents above identity on spaces and idempotents below identity on algebras.

The \href{https://ncatlab.org/nlab/show/Karoubi+envelope#InComponents}{\dfn{Karoubi envelope}} of a category $\mathcal C$ has for each idempotent arrow $c:C\to C$ in $\mathcal C$ an object $(C,c)$ with arrows $h:(C,c)\to (C',c')$ given by those $h:C\to C'$ such that $c'\circ h\circ c=h$. If $\mathcal C$ is order-enriched, the \dfn{lower Karoubi envelope} is the full subcategory of the Karoubi envelope spanned by objects $(C,c)$ such that $c\le\id_C$ while the \dfn{upper Karoubi envelope} is the full subcategory spanned by objects $(C,c)$ such that $c\ge\id_C$.

In the definition below, $\mathbb C$ is assumed to be a (strict) flat double category (Definitions~\ref{def:double-category} and~\ref{def:flat-ultraflat}): strictness lets us define idempotence by the equation $R\then R=R$, and flatness ensures that relations are preordered, so we may write $R\le\Id_C$ (ultraflatness would moreover make that order partial). The proarrow projection of $\Spl_\le(\mathbb C)$ is the lower Karoubi envelope of the proarrow projection of $\mathbb C$.

\begin{definition}\label{def:splitting-idempotents}
Let $\mathbb C$ be a flat double category.
The double category $\Spl_\le(\mathbb C)$ is defined by the following data.
\begin{itemize}
\item Objects $(C,R)$ where $C$ is an object of $\mathbb C$ and $R:C\looparrowright C$ is an idempotent proarrow in $\mathbb C$ below the identity ($R\le\Id_C$). 
\item Arrows $f:(C,R)\to(C',R')$ where $f$ is an arrow of $\mathbb C$ such that $(f,f):R\looparrowright R'$ is a square in $\mathbb C$, that is, $cRd\Rightarrow f(c) R' f(d)$.
\item Proarrows $S:(C,R)\looparrowright(C',R')$ where $S:C\looparrowright C'$ in $\mathbb C$ with $R \then S \then R' \,=\, S$. The identity proarrow on $(C,R)$ is $R$.
\item Squares $(f,g):S\looparrowright S'$ where $(f,g):S\looparrowright S'$ is a square in $\mathbb C$, that is, $cSd \Rightarrow f(c) S' g(d)$.
\end{itemize}
$\Spl_\ge(\mathbb C)$ is defined analogously but with respect to idempotents above the identity ($R\ge\Id_C$).
\end{definition}

\begin{remark} We note the following.
\begin{itemize}
\item An idempotent proarrow above (below) the identity is a monad (comonad) in the $\Pos$-category of relations.
\item Splitting idempotents preserves equivalence of double categories.
\item $\Spl_\ge(\mathbb C)=(\Spl_\le(\mathbb C^\op))^\op$ (on the proarrow projection this is the $^\co$ of Theorem~\ref{thm:duality-for-relations}).
\end{itemize}
\end{remark}

The next theorem follows directly from the definitions. It is the hinge around which this paper revolves. The fact that its proof is trivial is one of the reasons we employ the machinery of order-enriched and double categories. Its combinatorial ramifications will be spelt out in the remainder of the paper.

\begin{theorem}\label{thm:spl-pri-dl}
There is an equivalence of double categories
\[
    \Spl_\ge(\dRel(\Pri)) \simeq \Spl_\le(\dRel(\DL))^\op
\]
that extends the classical Priestley duality.
\end{theorem}

\begin{remark}
Using the notion of concretely order-regular category \cite{KMJ}, the theorem can be generalised by replacing the duality $\Pri \simeq \DL^\op$ by any order-enriched duality of concretely order-enriched categories.
\end{remark}

\begin{remark}\label{rem:flipped-splitting}
The reader by now will be curious about swapping the roles of upper and lower Karoubi constructions. That is, the machinery discussed here suggests we could just as well split idempotent $\DL$-relations above identity and idempotent Priestley relations below identity to obtain a dual version of Theorem~\ref{thm:spl-pri-dl}. While this works out formally, it does not yield anything interesting.

Consider $\DL$-preorders ${\sqsubseteq_A}$ and ${\sqsubseteq_B}$ on distributive lattices $A$ and $B$. These, effectively, are an order-theoretic way to encode lattice congruences on $A$ and $B$. Because the category $\DL$ is exact (is closed under quotients), a $\DL$-relation between $A$ and $B$ compatible with ${\sqsubseteq_A}$ and ${\sqsubseteq_B}$ amounts to a $\DL$-relation between the quotients $A/{\sqsubseteq_A}$ and $B/{\sqsubseteq_B}$. This can be used to show that in the double category $\Spl_\ge(\dRel(\DL))$, the proarrow projection is equivalent to $\Rel(\DL)$. Also, $\DL$ is reflective in the arrow projection of $\Spl_\ge(\dRel(\DL))$ in essentially the same way that posets form a reflective subcategory of the category of preorders.

On a Priestley space $X$, an idempotent ${\prec}$ below identity (corresponding to a congruence on the dual distributive lattice) determines a closed subspace consisting of the elements satisfying $x\prec x$ (call these reflexive elements). Because a closed subspace of a Priestley space is a Priestley space, we obtain a dual fact about $\Spl_\le(\dRel(\Pri))$: the proarrow projection is equivalent to $\Rel(\Pri)$ and the arrow projection has $\Pri$ as a coreflective subcategory.
\end{remark}

\section{Proximity lattices}\label{sec:proxlat}

The notion of a proximity lattice was introduced in the context of domain theory and stably compact spaces \cite{Smyth1992,JungSunderhauf1996}, and further developed by \cite{Moshier2004} for compact Hausdorff spaces. 

\begin{definition}\label{def:proximity-lattice}
  A \dfn{proximity lattice} $(A,\vee,\wedge, 0,1, \prec)$ is a distributive
  lattice $(A,\vee,\wedge, 0,1)$ with an idempotent $\DL$-relation $\prec$ below identity.
  The category $\Prox$ has arrows $f:(A,\prec)\to(A',\prec')$ which are $\DL$-morphisms and preserve the relation, that is, $a\prec b \Rightarrow \ f(a)\prec' f(b)$.
  The category $\ProxRel$ has arrows $R:(A,\prec)\looparrowright(A',\prec')$ such that ${\prec} \then R \then {\prec'} \,=\, R$.
  The double category $\dProx$ takes the arrows from $\Prox$ and the proarrows from $\ProxRel$. The squares are inherited from the underlying posets.
\end{definition}

\begin{remark}\label{rem:prox-axioms}
The conditions on a proximity relation $\prec$ are equivalent to the following:
  $$\prec\circ\prec {=} \prec \quad\quad {\prec}\subseteq\ \le$$ 
  \begin{gather*}
  \frac{a\le b \quad\quad b\prec c \quad\quad c\le d}{a\prec d}
  \end{gather*}
  \begin{gather*}
  0 \prec 0 \quad\quad 1\prec 1 \\[0.7em]
  \frac{a \prec b \quad\quad c\prec d}{a \vee c \, \prec \, b\vee d} \quad\quad \frac{a   \prec b \quad\quad c\prec d}{a \wedge c \, \prec \, b \wedge d}
   \end{gather*}
Similarly, the conditions on a relation (proarrow) $R:(A,\prec)\looparrowright(A',\prec')$ are equivalent to the following:
   \begin{gather*}
    \frac{a_1\le a_2 \quad\quad a_2 R\,  a_1' \quad\quad a_1'\le a_2'}{a_1 R\,  a_2'} \\[0.7em]
 0 R 0' \quad\quad 1 R 1' \\[0.7em]
  \frac{a R b \quad\quad cR d}{(a \vee c) \, R \, (b\vee d)} \quad\quad \frac{a   R b \quad\quad cR d}{(a \wedge c) \, R \, (b\wedge d)}\\[0.7em]
    \frac{a_1\prec a_2 \quad\quad a_2 R\, a_1' \quad\quad a_1'\prec a_2'}{a_1 R\, a_2'}   \end{gather*}  
\end{remark}

\begin{remark}\label{rem:prox-simplified-axioms}
Given that $\prec$ is a weakening relation below identity on the underlying poset of a $\DL$, it is enough to check
   \begin{gather*}
   a\prec c \ \textrm{ only if there is $b$ such that } \ a\prec b \ \textrm{ and } \ b\prec c  \\[0.7em]
   \frac{a \prec b \quad\quad b\prec c}{a \prec c}  \\[0.7em]
  0 \prec 0 \quad\quad 1\prec 1 \\[0.7em]
  \frac{a \prec b \quad\quad c\prec b}{a \vee c  \prec \, b} \quad\quad \frac{a   \prec b \quad\quad a\prec d}{a  \prec b \, \wedge d}
   \end{gather*}
in order to obtain a proximity lattice.
In particular, one derives the full rules for conjunction and disjunction of Remark~\ref{rem:prox-axioms} by weakening. Conversely, the last two rules above follow from those of Remark~\ref{rem:prox-axioms} using that in a $\DL$ we have $b\vee b=b$ and $a\wedge a = a$. The conditions on $R$ for $\wedge$ and $\vee$ can be simplified in an analogous way.

Moreover, since we know that $\prec$ and $\prec'$ are idempotent, the condition ${\prec} \then R \then {\prec'} \,=\, R$ is equivalent to ${\prec} \then R \,=\, R$ and $R \then {\prec'} \,=\, R$.
\end{remark}

The following proposition is immediate from the definitions. 
\begin{proposition}\label{prop:dprox-spl}
  $\dProx\cong\Spl_\le(\dRel(\DL))$ is an isomorphism of double categories. 
\end{proposition}

\begin{remark}
\begin{itemize}
  \item Henceforth, we will take $\dProx = \Spl_\le(\dRel(\DL))$ as the definition of (the double category of) proximity lattices. The categories $\Prox$ (proximity lattices with functional arrows) and $\ProxRel$ (proximity lattices with relational arrows) then arise as the arrow, respectively proarrow, projections of $\dProx$.
  \item While $\ProxRel$ is the proarrow projection of $\dProx$, it is not $\Rel(\Prox)$. The natural forgetful functor $\Prox\to\Pos$ sends $(A,\prec)$ to $(A,\le)$, since $\prec$ need not be reflexive, so identities in $\Rel(\Prox)$ are lattice orders, whereas identities in $\ProxRel$ are the proximities $\prec$.
  \item Going back to Section~\ref{sec:dualities-relations}, $\Prox$ is the full subcategory of $\SubDL$ consisting of the subordination algebras where the subordination is idempotent and below identity. 
\end{itemize}
\end{remark}

The paradigmatic example is the proximity lattice of a compact ordered space:

\begin{definition}\label{def:kohaus-proximity-lattice}
The proximity lattice $(A, \prec)$ is \dfn{the proximity lattice of a compact ordered space} $X$ if
   \begin{itemize}
     \item $A$ is the set of pairs $(d, u)$ where $d$ is a lower open set and $u$ is an upper open set and $d \cap u = \emptyset$;
     \item $(d, u) \le (d', u')$ if and only if $d \subseteq d'$ and $u \supseteq u'$.
     \item $(d, u) \prec (d', u')$ if and only if $u \cup d'$ covers $X$, ie, $u \cup d' = X$.
   \end{itemize}
 \end{definition}

\begin{remark}\label{rem:jung-sunderhauf-OK}
 Up to notation, this is as in Jung and S\"underhauf \cite[Thm.~23]{JungSunderhauf1996} who start from a stably compact space and work with pairs $(O,K)$ where $O$ open, $K$ compact saturated, $O\subseteq K$, and $(O,K)\prec(O',K')$ iff $K\subseteq O'$. Under the equivalence between stably compact spaces and compact ordered spaces (see e.g.\ \cite[Thm.~9.1.32]{Goubault-LarrecqNonHausdorff}, \cite[Prop.~VI-6.23]{GierzEtAl}), the mapping $(O,K)\mapsto(X\setminus K,\,O)$ sends JS-tokens to our pairs $(d,u)$, thereby reversing the order of the lattices.
\end{remark}

\begin{proposition}\label{prop:kohaus-to-prox}
  Mapping a compact ordered space $X$ to its proximity lattice defines a contravariant functor $\KOHaus \to \Prox$.
\end{proposition}
\begin{proof}
  First, we verify that $(A, \prec)$ as defined is a proximity lattice. 
  The set $A$ of pairs $(d, u)$ with $d$ a lower open, $u$ an upper open, and $d \cap u = \emptyset$ forms a distributive lattice (with meets and joins built from intersections and unions). 
  The relation $\prec$ is transitive and below the identity (componentwise order), and satisfies the required axioms: if $(d, u) \prec (d', u')$ and $(d', u') \prec (d'', u'')$, then $u \cup d' = X$ and $u' \cup d'' = X$, which implies $u \cup d'' = X$ (since $u' \subseteq u$ and $d' \subseteq d''$ by the order structure), so $(d, u) \prec (d'', u'')$. 
  To show that $\prec$ is interpolative, suppose $(d,u)\prec(d'',u'')$, ie $u\cup d''=X$.
  Set $K:=X\setminus u$ and $L:=X\setminus d''$. Then $K$ is a closed down-set, $L$ is a closed up-set, and $K\cap L=\emptyset$. 
  By Nachbin separation~\eqref{eq:nachbin-separation}, there exist a lower open $d'$ and an upper open $u'$ with $K\subseteq d'$, $L\subseteq u'$, and $d'\cap u'=\emptyset$. Thus $(d',u')\in A$, and $u\cup d'\supseteq u\cup K=X$ together with $u'\cup d''\supseteq L\cup d''=X$, so $(d,u)\prec(d',u')\prec(d'',u'')$.
  The compatibility with lattice operations follows from the set-theoretic properties of unions and intersections.
  
  On morphisms: For a continuous monotone map $f : X \to Y$ in $\KOHaus$, define $f^* : (A_Y, \prec_Y) \to (A_X, \prec_X)$ by $f^*(d, u) = (f^{-1}(d), f^{-1}(u))$ for $(d, u) \in A_Y$. Since $f$ is continuous and monotone, $f^{-1}(d)$ is a lower open in $X$ and $f^{-1}(u)$ is an upper open in $X$. If $d \cap u = \emptyset$ in $Y$, then $f^{-1}(d) \cap f^{-1}(u) = f^{-1}(d \cap u) = \emptyset$ in $X$, so $f^*(d, u) \in A_X$. The proximity relation is preserved: if $(d, u) \prec_Y (d', u')$, meaning $u \cup d' = Y$, then $f^{-1}(u) \cup f^{-1}(d') = f^{-1}(u \cup d') = f^{-1}(Y) = X$, so $f^*(d, u) \prec_X f^*(d', u')$. 
\end{proof}

To illustrate the power of relational morphisms, we show that the proximity lattice of the Priestley space (as a compact ordered space) of a distributive lattice $A$ is relationally isomorphic to $A$. 

\begin{remark}\label{rem:proxrel-iso-priestley}
Let $X$ be a Priestley space, $A$ its dual distributive lattice (the lattice of clopen downsets), and $A'$ the proximity lattice of the compact ordered space $X$. Then $\phi : A\to A'$ given by $\phi(a)=(a,X\setminus a)$ is an embedding but not an isomorphism in $\Prox$. It does, however, induce a relational isomorphism $A\cong A'$ in $\ProxRel$: writing $aRa'=\phi(a)\prec a'$ and $a'Sa=a'\prec \phi(a)$ for the relations associated with $\phi$, it is not difficult to verify that $R\then S={\le_A}$ (using interpolation of $\prec$) and $S\then R={\prec}$ (using Priestley separation). Thus $A$ and $A'$ are isomorphic as objects of $\ProxRel$ (but not of $\Prox$).
\end{remark}

\begin{example}\label{ex:unit-interval-proximity-lattice}
Let $X$ be the unit interval $[0,1]$ of real numbers with the natural topology and order. 
The proximity lattice $(A, \prec)$ of $X$ is the distributive lattice of pairs $(d, u)$ where $d$ is a lower open (of the form $[0,x)$ or $[0,1]$) and $u$ is an upper open (of the form $(y,1]$ or $[0,1]$) with $d \cap u = \emptyset$. 
The proximity relation $\prec$ is the relation specified by $([0,x'),(y',1])\prec([0,x),(y,1])$ if $y'<x$, see the top and middle picture in Figure~\ref{fig:proximity-lattice-unit-interval}. See also Example~\ref{ex:unit-interval-priestley-preorder}.
\end{example}

\begin{figure}[htbp]
  \centering
  \begin{tikzpicture}[scale=3.8]
  \useasboundingbox (-0.08,-0.06) rectangle (1.25,1.25);
  \draw[->] (0,0) -- (1.2,0) node[right] {$x$};
  \draw[->] (0,0) -- (0,1.2) node[above] {$y$};
  \fill[gray!10] (0,0) -- (1,1) -- (0,1) -- cycle;
  \draw[red, line width=1.5pt] (0,0) -- (1,1);
  \draw[black] (0,1) -- (1,1);
  \draw[red, line width=1.5pt] (0,0) -- (0,1);
  \coordinate (K) at (0.4, 0.7);
  \coordinate (A) at (0, 0.7);
  \coordinate (B) at (0.4, 0.4);
  \draw[blue, dashed] (K) -- (A);
  \draw[blue, dashed] (K) -- (B);
  \filldraw[black] (K) circle (0.5pt) node[above, yshift=2pt] {$([0,x), (y,1])$};
  \filldraw[blue] (A) circle (0.5pt) node[left] {$(\emptyset, (y,1])$};
  \filldraw[blue] (B) circle (0.5pt) node[right, xshift=2pt] {$([0,x), (x,1])$};
  \filldraw[black] (0,0) circle (0.5pt) node[below left] {$(0,0)$};
  \filldraw[black] (0,1) circle (0.5pt) node[left] {$(0,1)$};
  \filldraw[black] (1,1) circle (0.5pt) node[right] {$(1,1)$};
  \end{tikzpicture}

  \medskip
  \begin{tikzpicture}[scale=3.8]
  \useasboundingbox (-0.08,-0.06) rectangle (1.25,1.25);
  \draw[->] (0,0) -- (1.2,0);
  \draw[->] (0,0) -- (0,1.2);
  \fill[gray!10] (0,0) -- (1,1) -- (0,1) -- cycle;
  \draw[black] (0,0) -- (1,1) -- (0,1) -- cycle;
  \def\mx{0.4}
  \def\my{0.7}
  \fill[green!35] (0,0) -- (0,\mx) -- (\mx,\mx) -- (0,0) -- cycle;
  \draw[black, dashed] (\mx,\my) -- (\mx,\mx);
  \filldraw[black] (\mx,\my) circle (0.5pt) node[above right] {$(x,y)$};
  \draw[black] (\mx,\mx) circle (0.5pt) node[below right] {$(x,x)$};
  \filldraw[black] (0,0) circle (0.5pt);
  \filldraw[black] (0,1) circle (0.5pt);
  \filldraw[black] (1,1) circle (0.5pt);
  \end{tikzpicture}

  \medskip
  \begin{tikzpicture}[scale=3.8]
  \useasboundingbox (-0.08,-0.06) rectangle (1.25,1.25);
  \draw[->] (0,0) -- (1.2,0);
  \draw[->] (0,0) -- (0,1.2);
  \fill[gray!10] (0,0) -- (1,1) -- (0,1) -- cycle;
  \draw[black] (0,0) -- (1,1) -- (0,1) -- cycle;
  \def\myx{0.4}
  \def\myy{0.6}
  \begin{scope}[opacity=0.5]
    \fill[magenta] (0,\myy) -- (\myy,\myy) -- (1,1) -- (0,1) -- (0,\myy) -- cycle;
  \end{scope}
  \begin{scope}[opacity=0.5]
    \fill[blue] (\myx,\myx) -- (\myx,1) -- (1,1) -- (\myx,\myx) -- cycle;
  \end{scope}
  \draw[blue, line width=1pt] (\myx,\myx) -- (\myx,1);
  \draw[magenta, line width=1pt] (0,\myy) -- (\myy,\myy);
  \draw[green!70!black, line width=1pt] (\myy,\myy) -- (\myy,1) -- (1,1) -- (\myy,\myy);
  \node[blue, below] at (\myx, -0.02) {$x$};
  \node[magenta, left] at (-0.02, \myy) {$y$};
  \node[blue, font=\scriptsize] at (0.72, 0.48) {$F_x$};
  \node[magenta, font=\scriptsize] at (0.22, 0.82) {$G_y$};
  \node[green!70!black, font=\scriptsize] at (0.72, 0.85) {$\Box G_y$};
  \filldraw[black] (0,0) circle (0.5pt);
  \filldraw[black] (0,1) circle (0.5pt);
  \filldraw[black] (1,1) circle (0.5pt);
  \end{tikzpicture}

  \caption{\textbf{Top:} Encoding of the proximity lattice; 
  the order increases up and right; 
  join-irreducibles on the red lines; the black point is the join (coordinate-wise max) of the two blue points. \textbf{Middle:} A point $(x,y)$ and the set $\{(x',y') \mid (x',y')\prec (x,y)\}$ (green). \textbf{Bottom:} Prime filters $F_x$ (blue, vertical) and $G_y$ (magenta, horizontal). The round prime filter $\Box F_x$ is the interior of $F_x$. $\Box G_y$ is indicated. $x\le  y \ \Rightarrow \ \Box G_y \subseteq \Box F_x$.}
  \label{fig:proximity-lattice-unit-interval}
\end{figure}
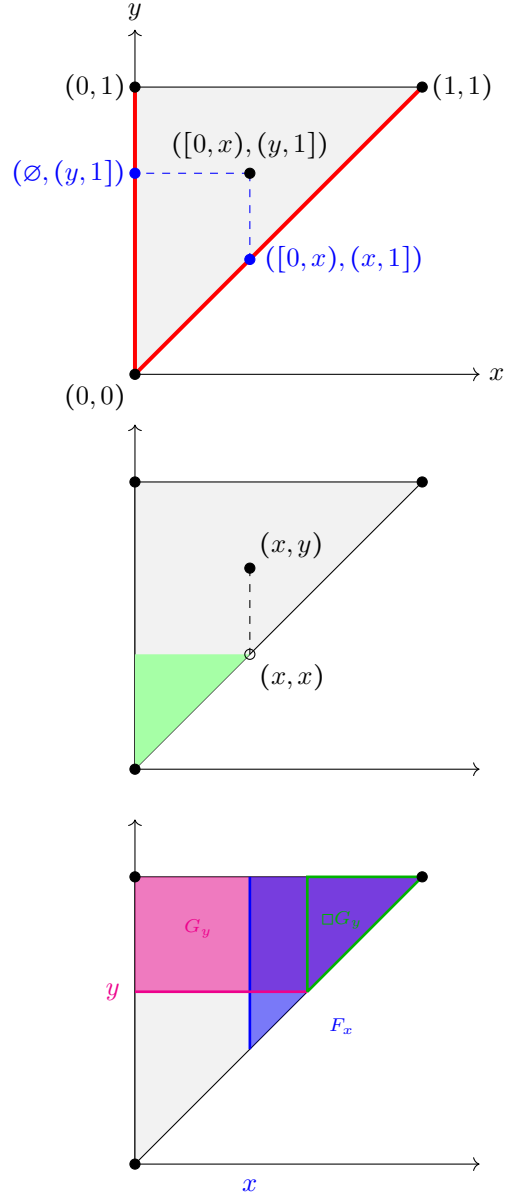

Proximity lattices that arise from compact ordered spaces have a stronger interpolation property; see eg \cite[Def~3.1.10]{Kegelmann1999}. (For the relationship of our proximity lattices to the strong proximity lattices of \cite{JungKegelmannMoshier1999}, see also Section~\ref{sec:adjoint-relations-proximity-lattices}.)

\begin{definition}\label{def:str-proximity-lattice}
A proximity lattice is \dfn{$\vee$-strong} if it satisfies
$$ x\prec y\vee z \ \Rightarrow \ \exists y'\prec y. \exists z'\prec z.x\prec y'\vee z'$$
and it is \dfn{$\wedge$-strong} if it satisfies
$$ y\wedge z\prec x \ \Rightarrow \ \exists y'\succ y. \exists z'\succ z. y'\wedge z'\prec x.$$
A proximity lattice is \dfn{strong} if it is both $\vee$-strong and $\wedge$-strong.
\end{definition}

One reason to isolate strongness is that it guarantees a Gentzen-style cut rule:
\begin{proposition}\label{prop:gentzen-cut}
  If $(A,\prec)$ is $\vee$-strong or $\wedge$-strong, then it satisfies the Gentzen cut rule
  \begin{gather*}
  \frac{a\prec b\vee c \quad\quad a\wedge c\prec b}{a\prec b}.
  \end{gather*}
\end{proposition}
\begin{proof}
  We assume that $A$ is $\vee$-strong, and that $a\prec b\vee c$ and $a\wedge c\prec b$. By $\vee$-strongness there are $b'\prec b$ and $c'\prec c$ with $a\prec b'\vee c'$. Since ${\prec}\subseteq{\le}$ we have $a\le b'\vee c'$, hence by distributivity
  \[
  a \;=\; a\wedge(b'\vee c') \;=\; (a\wedge b')\vee(a\wedge c').
  \]
  Weakening gives $a\wedge b'\prec b$ (from $b'\prec b$) and $a\wedge c'\prec b$ (from $a\wedge c\prec b$ and $c'\prec c$). The join rule for $\prec$ yields $a\prec b$.
\end{proof}

The following proposition due to Kegelmann \cite[Prop~3.1.34]{Kegelmann1999} shows that the proximity lattice of a compact ordered space is strong.

\begin{proposition}\label{prop:kohaus-prox-strong}
The proximity lattice of a compact ordered space is strong.
\end{proposition}

\section{Priestley preorders}\label{sec:priestley-preorders}

This section introduces Priestley preorders, the dual counterpart to the proximity lattices of Section~\ref{sec:proxlat}. We define the quotient of a Priestley preorder by its symmetric kernel and show that it yields a compact ordered space (Definition~\ref{def:compact-ordered-space}; see also \eqref{eq:nachbin-separation} and \cite{AlvarezManillaEtAl2004}).

\begin{definition}\label{def:priestley-preorder}
  A \dfn{Priestley preorder} $(X, \sqsubseteq)$ is a Priestley space $X$ with an idempotent $\Pri$-relation $\sqsubseteq$ above the identity. The category $\PriPre$ of Priestley preorders has arrows $f:(X,\sqsubseteq)\to(X',\sqsubseteq')$ which are continuous monotone functions preserving the preorder: $x \sqsubseteq y\ \Rightarrow \ f(x) \sqsubseteq' f(y)$.
  The double category $\dPriPre$ has proarrows $S:(X,\sqsubseteq)\looparrowright(X',\sqsubseteq')$ such that $S:X\looparrowright X'$ is a $\Pri$-relation and ${\sqsubseteq} \then S \then {\sqsubseteq'} \,=\, S$.
\end{definition}

\begin{remark}
  For a Priestley preorder $(X, \sqsubseteq)$ we have:
  \begin{enumerate}
  \item ${\sqsubseteq}$ is weakening: $x \le_X x' \ \wedge \ x' \sqsubseteq y' \ \wedge \ y' \le_X y \ \Rightarrow \ x \sqsubseteq y$
  \item ${\sqsubseteq}$ is closed in the product topology on $X \times X$
  \item ${\sqsubseteq}$ is above the identity: ${\le_X} \, \subseteq \ {\sqsubseteq}$
  \item ${\sqsubseteq}$ is idempotent: ${\sqsubseteq} \then {\sqsubseteq} \,=\, {\sqsubseteq}$
  \end{enumerate}
  To conclude the remark, let us emphasise that $\sqsubseteq$ is indeed a preorder.
\end{remark}

Given a Priestley preorder $(X, \sqsubseteq)$, let ${\equiv}$ be the equivalence relation given by $x \equiv y$ iff $x \sqsubseteq y$ and $y \sqsubseteq x$, and let $[x]$ be the equivalence class of $x$.
\begin{definition}\label{def:quotient-functor}
  The \dfn{quotient} of a Priestley preorder $(X, \sqsubseteq)$ is the quotient space $X/{\equiv}$ equipped with the partial order given by $[x] \leq [y]$ iff $x \sqsubseteq y$. We denote this quotient by $X/{\sqsubseteq}$.
\end{definition}

\begin{proposition}\label{prop:quotient-to-kohaus}
  The quotient construction of Definition~\ref{def:quotient-functor} defines a functor
  \[
    \PriPre \to \KOHaus
  \]
  from the category of Priestley preorders to the category of compact ordered spaces. In particular, every quotient of a Priestley preorder is a compact ordered space.
\end{proposition}
\begin{proof}
  Since ${\sqsubseteq}$ is closed, so is ${\equiv} \,=\, {\sqsubseteq} \cap{\sqsupseteq}$. The quotient of a compact Hausdorff space by a closed equivalence relation is compact Hausdorff.
  The relation $[x]\le[y]\Leftrightarrow x\sqsubseteq y$ is a well-defined partial order on $X/{\equiv}$ (well-defined and transitive by idempotence of ${\sqsubseteq}$; antisymmetric by construction of ${\equiv}$), and it is closed as the image of the closed saturated set ${\sqsubseteq}$ under the quotient map $X\times X\to(X/{\equiv})^2$.
  On morphisms: if $f:(X,\sqsubseteq)\to(Y,\sqsubseteq')$ is continuous, monotone, and preserves ${\sqsubseteq}$, define $X/{\sqsubseteq}\to Y/{\sqsubseteq'}$ by $[x] \mapsto [f(x)]$.
\end{proof}

We will pick this up again in Theorem~\ref{thm:spec-quotient-kohaus}.

\begin{definition}\label{def:cover-functor}
  The \dfn{cover functor} $\KOHaus \to \PriPre$ is defined by composing the contravariant functor $\KOHaus \to \Prox$ from Proposition~\ref{prop:kohaus-to-prox} with the duality between proximity lattices and Priestley preorders.
\end{definition}

\begin{example}\label{ex:unit-interval-priestley-preorder}
Let $X$ be the unit interval $[0,1]$ as in Example~\ref{ex:unit-interval-proximity-lattice}. 
The Priestley preorder $(X', \sqsubseteq)$ dual to the proximity lattice $(A, \prec)$ of the unit interval is obtained by taking for $X'$ the space of prime filters $\Spec(A)$. 
As indicated at the top of Figure~\ref{fig:proximity-lattice-unit-interval}, the filters generated by $([0,x),(y,1])$ with $x<y<1$ are not prime because $([0,x),(y,1])=([0,x),\emptyset)\vee ([0,y),(y,1])$ is not join-irreducible. 
The preorder $\sqsubseteq$ is defined by $F \sqsubseteq F' \ \Leftrightarrow \ \forall (d,u)\in F \ \& \ (d,u)\prec (d',u') \Rightarrow (d',u')\in F'$. 
See also Example~\ref{ex:encoding-unit-interval-proximity-lattice}.
\end{example}

\section{Dedekind cuts as adjoint relations}\label{sec:dedekind-cuts}

In this section, we will work in the category where an object $(A,\prec)$ consists of a set $A$ with an idempotent relation ${\prec}\subseteq A\times A$, that is, ${\prec} \then {\prec} \,=\, {\prec}$. 
Arrows $R:(B,\ll)\to(A,\prec)$ are relations $R\subseteq B\times A$ satisfying ${\ll} \then R \then {\prec} \,=\, R$.

We will also consider the case where, moreover, the set $A$ carries the structure of a distributive lattice, where $\prec$ is a proximity relation and $R$ is a $\DL$-relation.

Consider the following diagram where $A,B$ are posets, and $\prec,\ll$ are idempotent relations.
\begin{equation}\label{eq:adjointrelsinpos}
\begin{tikzcd}[column sep=3em]
	(A,\prec) & \top & (B,\ll)
	\arrow["L"', curve={height=12pt}, from=1-1, to=1-3]
	\arrow["R"', curve={height=12pt}, from=1-3, to=1-1]
\end{tikzcd}
\end{equation}
We will assume that $L,R$ are relational morphisms between $(A,\prec)$ and $(B,\ll)$, that is, they are weakening-closed with respect to the idempotents
\begin{align}\label{eq:adjointrelsround}
  aLb & \ \Leftrightarrow \  \exists a',b'\,.\,a\prec a' \,L\, b' \ll b\\
  bRa & \ \Leftrightarrow \  \exists b',a'\,.\,b\ll b' \,R\, a' \prec a
\end{align}
and we have $L\dashv R$, that is,
\begin{align*}
  a\prec a' & \ \Rightarrow \ \exists b\,.\,aLbRa' & \textrm{(Unit)}\\
\exists a\,.\, bRaLb' & \ \Rightarrow \ b\ll b' & \textrm{(Counit)}
\end{align*}

We will write such adjoint pairs as arrows aligning with the left-adjoint:

$$(L\dashv R):(A,\prec)\looparrowright(B,\ll)$$

We write $Ra=\{b\mid bRa\}$ and $aL=\{b\mid aLb\}$ and similarly for $bR$ and $Lb$.

\begin{example}
  Take $A=1=\{a\}$ and let $\prec$ be the identity. Then $aL$ is an upset and $Ra$ is a downset. (Unit) says that their intersection is non-empty and (Counit) says that everything in $Ra$ is $\ll$-below everything in $aL$. If $\ll$ is anti-symmetric\footnote{which is always true in our setting due to $\ll$ being contained in $\le$} it follows that the intersection of $aL$ and $Ra$ contains exactly one element. Note that this can be used to pick out exactly the reflexive elements\footnote{An element $b$ is reflexive if $b\ll b$.} of $B$ via left-adjoints $1\to B$.
\end{example}

The next example will give rise to the notion of cut and of round ideal. 

\begin{proposition}\label{prop:pointsarecuts}
In the situation of Diagram \eqref{eq:adjointrelsinpos} take $B=\twobb=\{0<1\}$ and let ${\ll}={\le}$. Then $(L0,1R)$ is a \dfn{cut} in the sense that 
\begin{itemize}
\item $L0$ is a downset and $1R$ is an upset 
\item (roundness) $\forall a\in L0\,.\,\exists a'\in L0\,.\,a\prec a'$ and $\forall a\in 1R\,.\,\exists a'\in 1R\,.\,a'\prec a$
\item (disjointness) $L0\cap 1R=\emptyset$ 
\item (tightness) If $a\prec a'$ then $a\in L0 \textrm{ or } a'\in 1R$ 
\end{itemize}
and every cut is of this form for a unique pair $L\dashv R$.
If, moreover, $A$ is a lattice and $L,R$ are lattice relations, then $L0$ is an ideal and $1R$ is a filter.
\end{proposition}
\begin{proof} 
Given  $L\dashv R$, we prove that $(L0,1R)$ is a cut.
Roundness is immediate from \eqref{eq:adjointrelsround}. 
For disjointness observe that $1Ra$ and $aL0$ implies $1\le 0$. 
For tightness, assume $a\prec a'$, which implies that there is $b$ such that $aLb$ and $bRa'$. If $b=0$ then $a\in L0$ and if $b=1$ then $a'\in 1R$.  

For the converse, let $(D,U)$ be a cut. Then we define $1Ra\ \Leftrightarrow \ a\in U$ and $aL0 \ \Leftrightarrow \ a\in D$ and check that $0Ra \ \Leftrightarrow \ \exists a'\,.\, a'\prec a$ 
and $aL1 \Leftrightarrow \ \exists a'\,.\, a\prec a'$ are forced.

Moreover, if $a,a'\in L0$ then also $(a\vee a') L 0$ and if $a,a'\in 1R$ then $(a\wedge a') R 1$. 
\end{proof}

\begin{remark}
  If $A$ is a preorder, a cut is determined by a function $f:A\to\twobb$ such that $aL0 = \twobb(fa,0)$  and $1Ra = \twobb(1,fa)$.\footnote{Equivalently, $aL0 \iff fa= 0$ and $1Ra \iff fa= 1$.}
\end{remark}

\begin{example}
  The unit interval $[0,1]$ with the way-below\footnote{Defined here by $0\prec0$ and $x<y\Rightarrow x\prec y$ and $1\prec1$.} relation $\prec$, whether taken as a set of rational numbers or as a set of real numbers, has exactly one cut for each real number $r\in\mathbb R$ with $0\le r\le 1$.
\end{example}

\begin{definition}\label{def:round-filter}
If $A$ is a lattice and $L\dashv R:(A,\prec)\to(\twobb,\le)$ is an adjoint pair of lattice relations as in Proposition~\ref{prop:pointsarecuts}, we call a set of the form $L0$ a \dfn{round ideal} and a set of the form $1R$ a \dfn{round filter}. 
If $A$ is a proximity lattice, we denote by 
$\RFilt(A)$ 
the set of round filters of $A$. 
\end{definition}

The next proposition shows that on strong proximity lattices round ideals and filters are prime.

\begin{proposition}\label{prop:round-prime}
  If $A$ is a strong proximity lattice and $L\dashv R:(A,\prec)\to(\twobb,\le)$ is an adjoint pair of $\DL$-relations, then $L0$ is a round prime ideal and $1R$ is a round prime filter. 
\end{proposition}
\begin{proof}
  Assume $a \vee b \in 1R$. By roundness, there exists $c \in 1R$ with $c \prec a \vee b$. By $\vee$-strongness, there exist $a', b'$ with $a' \prec a$, $b' \prec b$, and $c \prec a' \vee b'$, hence $a' \vee b'\in 1R$. Assume for a contradiction that $a \notin 1R$ and $b \notin 1R$. By tightness, $a'\in L0$ and $b'\in L0$. Since $L$ is a lattice relation, we have $a'\vee b'\in L0$. By disjointness, this contradicts $a'\vee b'\in 1R$. A similar argument shows that $L0$ is a prime ideal.
\end{proof}

The next proposition establishes a prime ideal theorem for strong proximity lattices. It shows that any partial cut can be extended to a tight cut.

\begin{proposition}\label{prop:prime-cut-theorem}
  In a strong proximity lattice $(A,\prec)$, if $I:A\looparrowright \twobb$ and $F:\twobb\looparrowright A$
  satisfy $F \then I \, \subseteq \, \mathord{\le}$, then there is an adjoint pair $(I^*,F^*):A\looparrowright \twobb$ so that $I\subseteq I^*$ and $F\subseteq F^*$.
\end{proposition}
\begin{proof}
  Because composition preserves directed joins, Zorn's Lemma yields a maximal pair $(I^*,F^*)$ above $(I,F)$ satisfying $F^* \then I^* \, \subseteq \, \mathord{\le}$. We claim that ${\prec} \, \subseteq \, I^* \then F^*$, so that $(I^*,F^*)$ is an adjoint pair.

  Suppose $c\prec c'$, but towards a contradiction assume that $c \mathrel{(I^* \then F^*)} c'$ fails. In particular, then either $c \mathrel{I^*} 0$ or $1 \mathrel{F^*} c'$ fails.
  By interpolation, pick $c''$ so that $c\prec c''\prec c'$. So
  both $c'' \mathrel{I^*} 0$ and $1 \mathrel{F^*} c''$ fail.
  Now consider the ideal $I^! = I^* \vee \downarrow c''$, and the filter $F^! = F^* \vee \uparrow c''$, both being proper extensions.

  By maximality, $F^* \then I^! \,=\, \twobb\times\twobb \,=\, F^! \then I^*$.
  That is, there are elements $d$ and $e$ so that $1 \mathrel{F^*} d \mathrel{I^!} 0$ and $1 \mathrel{F^!} e \mathrel{I^*} 0$.
  The joins defining $I^!$ and $F^!$ mean that there are also elements $x$ and $y$ so that $d\prec x\vee c''$, $x \mathrel{I^*} 0$, $c''\wedge y\prec e$ and $1 \mathrel{F^*} y$. But then $1 \mathrel{F^*} d\wedge y$ and $x\vee e \mathrel{I^*} 0$. By Gentzen cut (Proposition~\ref{prop:gentzen-cut}) with cut element $c''$, we have $d\wedge y\prec x\vee e$. So $1 \mathrel{(F^* \then I^*)} 0$, contradicting $F^* \then I^* \, \subseteq \, \mathord{\le}$.
\end{proof}

\begin{remark}[Cuts vs.\ prime cuts]\label{rem:cuts-vs-prime-cuts}
The distinction between a general cut and a \dfn{prime cut} (ie, a round prime filter) arises from the type of adjoint pair $L\dashv R$ used to define them:
\begin{itemize}
\item General cuts arise from adjoint pairs $L\dashv R:(A,\prec)\to(\twobb,\le)$ of weakening relations. These give a notion of point analogous to Dedekind cuts: a partition of $A$ into a round lower part $L0$ and a round upper part $1R$ satisfying disjointness and tightness.
\item Prime cuts arise when $A$ is a strong proximity lattice and $L$ and $R$ are $\DL$-relations, see Proposition~\ref{prop:round-prime}.
\end{itemize}
For Stone and Priestley duality, points are given by $\DL$-morphisms $A\to\twobb$. For compact ordered spaces, points are given by left-adjoint $\ProxRel$-relations $(A,\prec)\looparrowright\twobb$.
\end{remark}

\section{Round filters}\label{sec:round-filters}

This section continues from the definition of round filters in Definition \ref{def:round-filter} and studies the notion of roundification. In terms of our example of Cantor space in Section~\ref{sec:cantor}, the aim is to understand when a quotient of a Priestley preorder identifies two prime filters. 

We recall that the dual of the proximity $\prec$ of a proximity lattice $A$ is the preorder $\sqsubseteq$ on $\Spec(A)$ given by 

$$F \sqsubseteq F' \ \Longleftrightarrow \ \forall a,a'\in A.\, (a\prec a' \ \& \ a\in F \ \Rightarrow \ a'\in F').$$

\begin{definition}\label{def:equiv-on-spec}
The equivalence relation $\equiv$ on $\Spec(A)$ is defined by ${\equiv} \,=\, {\sqsubseteq} \cap{\sqsupseteq}$, ie,
\[
\forall a,a'\in A.\; a\prec a' \ \Rightarrow\ (a\in F' \ \Rightarrow \ a'\in F)\ \& \ (a\in F \ \Rightarrow \ a'\in F').
\]
\end{definition}

The aim of the next few propositions is to show that two prime filters are equivalent if and only if they have the same roundification.

Recall from Definition~\ref{def:round-filter} that a round filter $F$ is a filter satisfying $a\in F\Rightarrow\exists a'\in F.\, a'\prec a$, or, in the notation of the next definition, $F=\Box F$.

\begin{definition}\label{def:box-roundification}
For a filter $F$ of a proximity lattice $A$, its \dfn{roundification} $\Box F$ is
\[
\Box F = \{a\in F \mid \exists a'\in F.\, a'\prec a\}.
\]
\end{definition}

$\Box$ is a comonad on the poset of filters:

\begin{proposition}\label{prop:box-is-comonad}
  $\Box F \subseteq F$, $\Box F \subseteq \Box\Box F$ 
\end{proposition}
\begin{proof}
  For $\Box F \subseteq \Box\Box F$, assume $a\in \Box F$. By definition $\exists a'\in F.\, a'\prec a \ \& \ a\in F$. By interpolation, we have $\exists a',a''\in F.\, a'\prec a''\prec a \in F$, which implies $\exists a''\in\Box F.\, a''\prec a \ \& \ a\in \Box F$, hence $a\in \Box\Box F$.
\end{proof}

$\Box F$ is the largest round filter contained in $F$:

\begin{proposition}\label{prop:box-is-filter}
  $\Box F$ is a filter.
\end{proposition}
\begin{proof}
\emph{Non-empty:} Since $1 \prec 1$ and $1 \in F$, we have $1 \in \Box F$. \emph{Upward closed:} Follows from $\Box F \subseteq F$ and $\prec$ being weakening. \emph{Closed under meets:} Follows from $a'\prec a \ \& \ b'\prec b \ \Rightarrow \ a'\wedge b'\prec a\wedge b$ and $F$ being a filter.
\end{proof}

Alongside the Priestley spectrum $\Spec(A)$ of prime filters of a distributive lattice $A$, we make the following definitions.
\begin{definition}\label{def:filt-rfilt-rpfilt}
For a lattice $A$, we denote by $\Filt(A)$ the lattice of filters of $A$.
For a proximity lattice $A$, we write $\RFilt(A)$ for the lattice of round filters and $\RPFilt(A)$ for the set of round prime filters.
\end{definition}

\begin{proposition}\label{prop:box-coreflection}
$\Box : \Filt(A) \to \RFilt(A)$ is right-adjoint to the inclusion $\RFilt(A) \to \Filt(A)$, which preserves finite meets.
\end{proposition}
\begin{proof}
The adjointness property follows from Propositions \ref{prop:box-is-comonad} and \ref{prop:box-is-filter} showing that $\Box$ is a comonad on $\Filt(A)$ factoring through $\RFilt(A)$. To show that the inclusion preserves finite meets, let $R,R'$ be two round filters. 
We have to show that the filter $R\cap R'$ is round. 
Let $a\in R\cap R'$. 
By roundness of $R$ and $R'$, there exist $b\in R$ with $b\prec a$ and $b'\in R'$ with $b'\prec a$. 
Then $b\vee b'\prec a$ by Remark~\ref{rem:prox-simplified-axioms}, and $b\vee b'\in R\cap R'$ since filters are upward closed.
\end{proof}

\begin{remark}\label{rem:box-locale}
It follows that the inclusion $\RFilt(A) \to \Filt(A)$ is a frame homomorphism and that $\Box$ is the corresponding locale map. In particular, the join and finite meet of round filters are given by their filter-join and filter-meet (intersection).
\end{remark}

\begin{corollary}\label{cor:box-preserves-meets}
  $\Box$ preserves finite meets: $\Box(F\cap F')=\Box F\cap\Box F'$. In particular, $\Box$ is an interior operator on the frame of filters.
\end{corollary}
\begin{proof}
$\Box$ is a comonad by Proposition~\ref{prop:box-is-comonad} and preserves finite meets by Proposition~\ref{prop:box-coreflection}.
\end{proof}

\begin{proposition}\label{prop:box-equiv}
$\Box F = \Box F' \ \Leftrightarrow \ F \equiv F'.$
\end{proposition}
\begin{proof}
For $\Rightarrow$, assume $\Box F = \Box F'$. We show $F \sqsubseteq F'$ (and $F' \sqsubseteq F$ follows symmetrically). Suppose $a \prec a'$ and $a \in F$. We need $a' \in F'$. Since $a \in F$ and $a \prec a'$, we have $a' \in \Box F$ by definition. But $\Box F = \Box F'$, so $a' \in \Box F' \subseteq F'$.

For $\Leftarrow$, assume $F \equiv F'$. We show $\Box F \subseteq \Box F'$ (and the reverse inclusion follows symmetrically). Suppose $a \in \Box F$, so $a \in F$ and there exists $a' \in F$ with $a' \prec a$. By $F \sqsubseteq F'$, since $a' \prec a$ and $a' \in F$, we have $a \in F'$. It remains to find $b \in F'$ with $b \prec a$. By idempotence of $\prec$, from $a' \prec a$ we obtain $a'' \in A$ with $a' \prec a'' \prec a$. Since $a' \in F$ and $a' \prec a''$, by $F \sqsubseteq F'$ we have $a'' \in F'$. Thus $a'' \in F'$ and $a'' \prec a$, hence $a \in \Box F'$.
\end{proof}

The next proposition shows that for the proximity lattices that arise as duals of compact ordered spaces, the roundification of a prime filter is again a prime filter.

\begin{proposition}\label{prop:box-prime-strong}
  If $F$ is a prime filter of a $\vee$-strong proximity lattice $A$, then $\Box F$ is a prime filter.
\end{proposition}
\begin{proof}
Assume $a \vee b \in \Box F$. Then $a \vee b \in F$ and there exists $c \in F$ with $c \prec a \vee b$. By $\vee$-strongness, there exist $a', b'$ with $a' \prec a$, $b' \prec b$, and $c \prec a' \vee b'$. Since $c \in F$ and $F$ is upward closed, we have $a' \vee b' \in F$. Since $F$ is prime, either $a' \in F$ or $b' \in F$. If $a' \in F$, then $a \in \Box F$. Similarly if $b' \in F$, then $b \in \Box F$.
\end{proof}

\begin{example}\label{ex:encoding-unit-interval-proximity-lattice}
  Continuing from Example~\ref{ex:unit-interval-priestley-preorder}, we recall that Figure~\ref{fig:proximity-lattice-unit-interval} gives an encoding of the proximity lattice of the unit interval with its natural topology and order. The join-irreducibles are on the red lines of the top picture. The black point in the middle is the join of the two blue points. In the coordinate-wise encoding, the join is computed by taking the maximum in each coordinate. 
  
  Since the points $(x,x)$, $0 < x \le 1$, and $(0,y)$, $0 < y \le 1$, are join-irreducibles, their principal upsets, $F_x = \{ (x',y') \in A \mid x \le x', x \le y' \}$ and $G_y = \{ (x',y') \in A \mid y \le y', x \le y' \}$, respectively, are prime filters. 

  Using the definition of $\prec$ depicted in the middle picture, it is not difficult to check that $F_x$ and $G_x$ have the same roundification, that is, $\Box F_x = \Box G_x$, and, therefore, are identified when quotienting by the Priestley preorder dual to $\prec$. This quotient gives back the unit interval.
\end{example}

\begin{proposition}\label{prop:equiv-closed}
The equivalence relation $\equiv$ is closed in $\Spec(A) \times \Spec(A)$.
\end{proposition}
\begin{proof}
Since ${\equiv} \,=\, {\sqsubseteq} \cap{\sqsupseteq}$ and the inverse of a closed relation is closed, it suffices to show that $\sqsubseteq$ is closed. We show that the complement of $\sqsubseteq$ is open:
\begin{align*}
\{(F, F') \mid F \not\sqsubseteq F'\} &= \{(F, F') \mid \exists a, a'\in A.\, a \prec a' \,\&\, a \in F \,\&\, a' \notin F'\} \\
&= \bigcup_{a \prec a'} \bigl(\{F \mid a \in F\} \times \{F' \mid a' \notin F'\}\bigr)
\end{align*}
Each set $\{F \mid a \in F\} \times \{F' \mid a' \notin F'\}$ is open in the product of Priestley topologies (as a product of basic opens). Hence the complement of $\sqsubseteq$ is a union of opens, so $\sqsubseteq$ is closed.
\end{proof}

The prime round filters of a proximity lattice are the prime filters quotiented by the dual of the proximity:

\begin{theorem}\label{thm:spec-quotient-kohaus}
  Let $(A,\prec)$ be a strong proximity lattice and $(\Spec(A),\sqsubseteq)$ the dual Priestley preorder. Then $\Spec(A)/{\sqsubseteq}$ (Definition~\ref{def:quotient-functor}) is a compact ordered space, isomorphic (as an ordered topological space) to $\RPFilt(A)$ (Definition~\ref{def:filt-rfilt-rpfilt}). Moreover, the quotient map $q: \Spec(A) \to \Spec(A)/{\sqsubseteq}$ is continuous and monotone.
\end{theorem}
\begin{proof} 
  We show that the quotient map $q: \Spec(A) \to \Spec(A)/{\sqsubseteq}$ factors through an isomorphism with the space of round prime filters.

  Define $\phi: \Spec(A)/{\sqsubseteq} \to \RPFilt(A)$ by $\phi([F]) = \Box F$. This is well-defined and injective by Proposition~\ref{prop:box-equiv}. It is bijective onto the round prime filters because $\Box F$ is prime (Proposition~\ref{prop:box-prime-strong}) and round (Proposition~\ref{prop:box-is-comonad}), and every round filter equals its roundification (Definition~\ref{def:box-roundification}). 
  We equip $\RPFilt(A)$ with the subbasis $\{G \in \RPFilt(A) \mid a \in G\}$, $a \in A$,  and order $G \le G'$ iff $G \subseteq G'$. 
  
  The composite $\phi \circ q$ is continuous: The inverse image of $\{G \in \RPFilt(A) \mid a \in G\}$ is $\{F \in \Spec(A) \mid a \in \Box F\}$. Since $a \in \Box F$ iff $\exists a' \in F.\, a' \prec a$ and  since $a' \prec a\ \Rightarrow \ a' \le a$, we have
  \[
  \{F \in \Spec(A) \mid a \in \Box F\} = \bigcup_{a' \prec a} \{F \in \Spec(A) \mid a' \in F\}
  \]
  which is a union of basic Priestley opens, hence open. By the universal property of the quotient topology, $\phi$ is continuous.
  
  Define $[F] \sqsubseteq [F']$ in $\Spec(A)/{\sqsubseteq}$ iff $F \sqsubseteq F'$. This makes the quotient $q$ monotone. 
  
  $\phi$ is an order-isomorphism because $[F] \le [F']$ iff $\Box F \subseteq \Box F'$.
  
The space $\Spec(A)/{\sqsubseteq}$ is compact as the continuous image of a compact space. It is Hausdorff because ${\equiv} \,=\, {\sqsubseteq} \cap{\sqsupseteq}$ is a closed equivalence relation (Proposition~\ref{prop:equiv-closed}) on the compact Hausdorff space $\Spec(A)$. The order is closed since $\sqsubseteq$ is closed. 
\end{proof}

We show that every compact ordered space is recovered from the quotient of its cover. Under the inclusion order of Theorem~\ref{thm:spec-quotient-kohaus}, the recovered ordered space is the order dual $X^o$; equivalently, equipping round prime filters with reverse inclusion recovers $X$ itself.

\begin{definition}\label{def:cover-of-kohaus}
The \dfn{cover} of a compact ordered space $X$ is the Priestley preorder $(\Spec(A_X),\sqsubseteq)$ dual to the proximity lattice $(A_X,\prec)$ of $X$ (Definition~\ref{def:kohaus-proximity-lattice}).
\end{definition}

Just as a Priestley space is recovered by sending each point to the (prime) filter of clopen down-sets that contain it \cite[Thm.~11.24]{DaveyPriestley2002}, a compact ordered space is recovered by sending each point to the round prime filter of tokens whose lower open contains the point.
\begin{theorem}\label{thm:cover-recovers-kohaus}
Let $X$ be a compact ordered space and $(\Spec(A_X),\sqsubseteq)$ its cover. Then, as ordered topological spaces,
\[
  \Spec(A_X)/{\sqsubseteq} \; \;\cong\; \; \RPFilt(A_X) \; \;\cong\; X^o,
\]
where $\RPFilt(A_X)$ carries the inclusion order of Theorem~\ref{thm:spec-quotient-kohaus}.
\end{theorem}
\begin{proof}
The first isomorphism is Theorem~\ref{thm:spec-quotient-kohaus}, using that $A_X$ is strong (Proposition~\ref{prop:kohaus-prox-strong}). For the second, define $\eta:X\to\RPFilt(A_X)$ by
\[
  \eta(x) \;=\; \bigl\{\,(d,u)\in A_X \mid x\in d\,\bigr\}.
\]
We verify that $\eta(x)$ is a round prime filter, that $\eta$ is a homeomorphism, and that $\eta(x)\subseteq\eta(y)$ if and only if $y\le x$.

\emph{Filter.} The top $1=(X,\emptyset)$ lies in $\eta(x)$. If $(d,u)\le(d',u')$ and $x\in d$, then $x\in d'$. If $x\in d$ and $x\in d'$, then $x\in d\cap d'$, and $(d\cap d',\,u\cup u')$ is a token, so $\eta(x)$ is closed under binary meets.

\emph{Round.} Let $(d,u)\in\eta(x)$, so $x\in d$. The complement $L:=X\setminus d$ is a closed up-set with ${\downarrow} x\cap L=\emptyset$. By Nachbin separation~\eqref{eq:nachbin-separation} there is a token $(d',u')$ with ${\downarrow} x\subseteq d'$, $L\subseteq u'$, and $d'\cap u'=\emptyset$. Then $(d',u')\in\eta(x)$ and $u'\cup d\supseteq L\cup d=X$, so $(d',u')\prec(d,u)$.

\emph{Prime.} If $(d,u)\vee(d',u')=(d\cup d',\,u\cap u')$ lies in $\eta(x)$, then $x\in d\cup d'$, so one of the two tokens lies in $\eta(x)$.

\emph{Continuity.} The preimage of a subbasic open $\{G\mid (d,u)\in G\}$ is $\{x\mid x\in d\}=d$, which is open.

\emph{Order-embedding.} Lower opens are down-sets, so $x\le y$ implies $\eta(y)\subseteq\eta(x)$. Conversely, if $y\not\le x$, then ${\downarrow} x\cap{\uparrow} y=\emptyset$; Nachbin separation~\eqref{eq:nachbin-separation} yields a token $(d,u)$ with $x\in d$ and $y\in u$, hence $(d,u)\in\eta(x)\setminus\eta(y)$. Thus $\eta(x)\subseteq\eta(y)$ if and only if $y\le x$, so $\eta$ is injective and order-reversing.

\emph{Surjectivity.} Let $G$ be a round prime filter of $A_X$, and set
\[
  S \;=\; \bigcap\bigl\{\,X\setminus u \mid \exists d.\;(d,u)\in G\,\bigr\}.
\]
Each $X\setminus u$ is a closed down-set; the family is closed under finite intersections (meets of tokens join the upper components) and contains $X$ (since $1\in G$). Compactness gives $S\neq\emptyset$. Roundness puts $S$ inside every lower open of a token in $G$: if $(d,u)\in G$, choose $(d',u')\in G$ with $(d',u')\prec(d,u)$, whence $S\subseteq X\setminus u'\subseteq d$.
Thus $G\subseteq\eta(x)$ for every $x\in S$.

Moreover, if $S\subseteq d$ for a token $(d,u)$, then $(d,u)\in G$: finitely many of the closed sets $X\setminus u_i$ already miss $X\setminus d$, their meet $(d_*,u_*)$ lies in $G$ and satisfies $(d_*,u_*)\prec(d,u)$, and $G$ is an upset. 
In particular, $S$ is an irreducible closed down-set: if $S\subseteq A\cup B$ for closed down-sets $A,B$ with neither containing $S$, separate ${\uparrow} y$ from $A$ and ${\uparrow} z$ from $B$ for $y\in S\setminus A$ and $z\in S\setminus B$ to obtain tokens $(d,\emptyset)$ and $(d',\emptyset)$ with $S\subseteq d\cup d'$, whence $(d\cup d',\,\emptyset)\in G$ and primeness forces a contradiction. 
Every nonempty closed subset of a compact ordered space has a maximal element \cite[Prop.~VI-5.3]{GierzEtAl}, so irreducibility yields a unique maximal $x\in S$ with $S={\downarrow} x$.
Then $G\subseteq\eta(x)$, and if $x\in d$ then $S\subseteq d$, so $(d,u)\in G$. Thus $\eta(x)=G$.

Since $X$ is compact and $\RPFilt(A_X)$ is Hausdorff, the continuous bijection $\eta$ is a homeomorphism. Combined with the order calculation, $\eta$ exhibits $X^o\cong\RPFilt(A_X)$ as ordered topological spaces.
\end{proof}
\begin{remark}\label{rem:abbadini-parallel-eta}
The surjectivity argument for $\eta$ in the proof above is similar to Abbadini's \cite[Thm.~5.2, Prop.~5.28]{Abbadini2020}, which recovers points from morphisms into $[0,1]$.
\end{remark}

\section{Adjoint relations in Priestley preorders}\label{sec:adjoint-relations-in-priestley-preorders}

Recall from Section~\ref{sec:dualities-relations} that $\Spl_\ge(\dRel(\Pri))$ has as objects Priestley preorders $(X,\sqsubseteq)$, that is, $X$ is a Priestley space and ${\sqsubseteq} : X \looparrowright X$ is an idempotent $\Pri$-relation above the identity. Functional morphisms (arrows) are Priestley morphisms that preserve the preorder, and relational morphisms (proarrows) are $\Pri$-relations $S : X \looparrowright X'$ satisfying ${\sqsubseteq} \then S \then {\sqsubseteq}' \,=\, S$. 

In this section, we are interested in the order-enriched category of relational morphisms in $\Spl_\ge(\dRel(\Pri))$ that have a right adjoint. This category will turn out to be equivalent to the category $\KOHaus$ of compact ordered spaces, extending Theorem~\ref{thm:cover-recovers-kohaus} (every $\KOHaus$ is covered by a $\PriPre$) and Proposition~\ref{prop:quotient-to-kohaus} (every $\PriPre$ quotients to a $\KOHaus$).

\begin{definition}\label{def:map-spl-pri}
  The order-enriched category $\Map(\Spl_\ge(\dRel(\Pri)))$ has objects Priestley preorders $(X,\sqsubseteq_X)$, $(Y,\sqsubseteq_Y)$, etc., and arrows $(X,\sqsubseteq_X) \to (Y,\sqsubseteq_Y)$ that are left-adjoint relations (proarrows) in $\Spl_\ge(\dRel(\Pri))$. Relations are ordered by inclusion.
  In general, if $\mathbb{C}$ is a double category of relations, $\Map(\mathbb{C})$ denotes the order-enriched category of \emph{maps} in $\mathbb{C}$: its arrows are the left-adjoint proarrows of $\mathbb{C}$.
\end{definition}

\begin{remark}
  We spell out the conditions on an arrow $L:(X,\sqsubseteq_X) \to (Y,\sqsubseteq_Y)$ in $\Map(\Spl_\ge(\dRel(\Pri)))$. 
  \begin{enumerate}
    \item Weakening with respect to identities: ${\le_X} \then L \then {\le_Y} = L$.
    \item Priestley relation: $L$ is a closed subspace of $X \times Y$.
    \item Weakening with respect to idempotents: ${\sqsubseteq_X} \then L \then {\sqsubseteq_Y} \,=\, L$.
    \item \dfn{$L$ has a right adjoint} $R \subseteq Y \times X$ satisfying the laws 1--3 above as well as the adjunction laws: for all $x,x' \in X$ and $y,y' \in Y$,
    \begin{align*}
      x \sqsubseteq_X x' \ &\Rightarrow \ \exists y \in Y.\, x \, L \, y \ \wedge \ y \, R \, x' & \text{(Unit)}\\
      \exists x \in X.\, y \, R \, x \, L \, y' \ &\Rightarrow \ y \sqsubseteq_Y y' & \text{(Counit)}
    \end{align*} 
  \end{enumerate}
\end{remark}

The quotient construction of Definition~\ref{def:quotient-functor} assigns to each Priestley preorder $(X,\sqsubseteq_X)$ the compact ordered space $X/{\sqsubseteq_X}$. 

The main result of this section is that morphisms in $\Map(\Spl_\ge(\dRel(\Pri)))$ correspond exactly to continuous monotone maps between these quotients.

\begin{theorem}\label{thm:map-kohaus-pri}
  Let $(X,\sqsubseteq_X)$ and $(Y,\sqsubseteq_Y)$ be Priestley preorders. There is a bijection between
  \begin{enumerate}
    \item arrows $(L \dashv R):(X,\sqsubseteq_X) \to (Y,\sqsubseteq_Y)$ in $\Map(\Spl_\ge(\dRel(\Pri)))$, and
    \item continuous monotone maps $f: X/{\sqsubseteq_X} \to Y/{\sqsubseteq_Y}$ in $\KOHaus$.
  \end{enumerate}
\end{theorem}
\begin{proof}
  {(1) $\Rightarrow$ (2).} Given an adjoint pair $L \dashv R : (X,\sqsubseteq_X) \to (Y,\sqsubseteq_Y)$, we define $f: X/{\sqsubseteq_X} \to Y/{\sqsubseteq_Y}$ as $f([x]) = [y]$ where $y$ is any element such that $x L y R x$. 
  
  It follows from (Unit) and $\sqsubseteq_X$ being above the identity that such a $y$ exists and from (Counit) that any two such are equivalent. 
  
  Independence of choice of representative follows from monotonicity and the fact that the quotients are partially ordered. 
  
  For monotonicity, if $xLyRx$, $x'Ly'Rx'$ and $x\sqsubseteq_X x'$, then $y\sqsubseteq_Y y'$ because $yRx\sqsubseteq x'Ly'$ implies $y\sqsubseteq y'$, due to (Counit) and weakening with idempotents.
  
  For continuity, note that the lower sets of the quotient $X/{\sqsubseteq_X}$ correspond to subsets of $X$ that are downward closed with respect to $\sqsubseteq_X$; we call these $\sqsubseteq_X$-lower sets. Similarly for upper sets.
  $L^\circ[-]$ (the set-theoretic converse of $L$) sends $\sqsubseteq_Y$-lower sets to $\sqsubseteq_X$-lower sets, and $R[-]$ sends $\sqsubseteq_Y$-upper sets to $\sqsubseteq_X$-upper sets. 
  So the inverse image $f^{-1}(D)$ of a $\sqsubseteq_Y$-lower set $D$ is exactly $L^\circ[D]$, and the inverse image of a $\sqsubseteq_Y$-upper set is $R[U]$. Because $X$ and $Y$ are compact and $L^\circ$ and $R$ are closed, $L^\circ[-]$ and $R[-]$ send closed sets to closed sets (as projections of closed sets from compact spaces). So $f^{-1}$ sends closed $\sqsubseteq_Y$-lower/upper sets to closed $\sqsubseteq_X$-lower/upper sets. Since these form a subbasis of closed sets for the KOHaus topology, $f$ is continuous.

  {(2) $\Rightarrow$ (1).} Given a continuous monotone map $f: X/{\sqsubseteq_X} \to Y/{\sqsubseteq_Y}$, we construct an adjoint pair $L \dashv R$.
  \begin{align*}
    x \, L \, y \quad &\Longleftrightarrow \quad f([x]) \le [y] \text{ in } Y/{\sqsubseteq_Y}\\
    y \, R \, x \quad &\Longleftrightarrow \quad [y] \le f([x]) \text{ in } Y/{\sqsubseteq_Y}
  \end{align*}
  
  With these definitions, the quotient relations $\bar L$ of $L$ and $\bar R$ of $R$ are weakening relations satisfying $\bar L \dashv \bar R$ (because $\bar L=f_\diamond$ and $\bar R=f^\diamond$), from which it follows that $L$ and $R$ are weakening with respect to the idempotents satisfying $L\dashv R$.

  It remains to show that $L$ and $R$ are closed subsets of $X \times Y$. Let $q_X: X \to X/{\sqsubseteq_X}$ and $q_Y: Y \to Y/{\sqsubseteq_Y}$ be the quotient maps. The order $\le$ on $Y/{\sqsubseteq_Y}$ is closed (since the quotient is a compact ordered space). Now:
  \[
  L = \{(x,y) \in X \times Y \mid f(q_X(x)) \le q_Y(y)\} = (f \circ q_X, q_Y)^{-1}(\le)
  \]
  Since $(f \circ q_X, q_Y): X \times Y \to (Y/{\sqsubseteq_Y})^2$ is continuous, and $\le$ is closed, $L$ is closed. Similarly, $R = (q_Y, f \circ q_X)^{-1}(\le)$ is closed.
  \end{proof}

\begin{theorem}\label{thm:adj-splge-pri-equiv-kohaus}
  The category $\Map(\Spl_\ge(\dRel(\Pri)))$ of Priestley preorders and left-adjoint relations is equivalent to the category $\KOHaus$ of compact ordered spaces with continuous monotone maps.
\end{theorem}
\begin{proof}
  We show that mapping $(X,\sqsubseteq_X)$ to $X/{\sqsubseteq_X}$ as defined above is an equivalence of categories. It is functorial because the identity relation $\sqsubseteq_X$ maps to the identity function on $X/{\sqsubseteq_X}$ and composition of adjoint relations $(L_1 ; L_2 \dashv R_2 ; R_1)$ maps to composition of functions $f_2 \circ f_1$. It is fully faithful by Theorem~\ref{thm:map-kohaus-pri}. It is essentially surjective because every compact ordered space is isomorphic to the quotient of a Priestley preorder (Theorem~\ref{thm:cover-recovers-kohaus}).
\end{proof}

\begin{corollary}\label{cor:rel-iso-kohaus}
  Priestley preorders $(X,\sqsubseteq_X)$ and $(Y,\sqsubseteq_Y)$ are relationally isomorphic if and only if their quotients $X/{\sqsubseteq_X}$ and $Y/{\sqsubseteq_Y}$ are isomorphic in $\KOHaus$.
\end{corollary}

\section{Adjoint relations in proximity lattices}
\label{sec:adjoint-relations-proximity-lattices}

This section dualises the results of the previous section from spaces to algebras.

By Theorem~\ref{thm:adj-splge-pri-equiv-kohaus}, Priestley preorders with adjoint relations are equivalent to compact ordered spaces; by Theorem~\ref{thm:spl-pri-dl}, they are dual to proximity lattices; and by Proposition~\ref{prop:kohaus-prox-strong}, the proximity lattice of a compact ordered space is strong.
Taken together, this shows that, as relational categories, proximity lattices are equivalent to strong proximity lattices.
In this section we give a purely algebraic construction of the result that every proximity lattice is relationally isomorphic to a strong one (Theorem~\ref{thm:strongification}). 

Recall from Section~\ref{sec:proxlat} that $\Spl_\le(\dRel(\DL))$ has as objects pairs $(A,\prec)$ where $A$ is a distributive lattice and ${\prec} : A \looparrowright A$ is an idempotent $\DL$-relation below the identity, that is, $(A,\prec)$ is a proximity lattice. 
Functional morphisms (arrows) are $\DL$-morphisms preserving $\prec$, and relational morphisms (proarrows) are $\DL$-relations $S : A \looparrowright A'$ satisfying ${\prec} \mathrel{;} S \mathrel{;} {\prec'} \,=\, S$.

\begin{definition}\label{def:map-spl-dl}
  The order-enriched category $\Map(\Spl_\le(\dRel(\DL)))$ has as objects proximity lattices $(A,\prec)$, and as arrows $(A,\prec) \to (A',\prec')$ left-adjoint relations (proarrows) in $\Spl_\le(\dRel(\DL))$. Relations are ordered by inclusion.
\end{definition}

\begin{remark}\label{rem:map-spl-dl}
  We spell out the conditions on an arrow $L : (A,\prec) \to (A',\prec')$ in $\Map(\Spl_\le(\dRel(\DL)))$. 
  \begin{enumerate}
    \item \dfn{$L$ is a weakening relation:} for all $a,a' \in A$ and $b,b' \in A'$,
    \[
      a' \prec a \ \wedge \ a \, L \, b \ \wedge \ b \prec' b' \ \Rightarrow \ a' \, L \, b'.
    \]
    \item \dfn{$L$ is a $\DL$-relation:} $L$ can be tabulated in $\DL$ and is compatible with the lattice structure (see Section~\ref{sec:dedekind-cuts}).
    \item \dfn{$L$ is a proximity relation:}
    \[
      {\prec} \mathrel{;} L \mathrel{;} {\prec'} \,=\, L.
    \]
    \item \dfn{$L$ has a right adjoint} $R \subseteq A' \times A$ satisfying ${\prec'} \mathrel{;} R \mathrel{;} {\prec} \,=\, R$ and the adjunction laws: for all $a,a' \in A$ and $b,b' \in A'$,
    \begin{align*}
      a \prec a' \ &\Rightarrow \ \exists b \in A'.\, a \, L \, b \ \wedge \ b \, R \, a' & \text{(Unit)}\\
      \exists a \in A.\, b \, R \, a \, L \, b' \ &\Rightarrow \ b \prec' b' & \text{(Counit)}
    \end{align*}
  \end{enumerate}
\end{remark}

It now follows from Theorem~\ref{thm:adj-splge-pri-equiv-kohaus} and the duality of the respective double categories (Theorem~\ref{thm:spl-pri-dl}):

\begin{theorem}\label{thm:adj-splle-dl-equiv-kohaus}
  The category $\Map(\Spl_\le(\dRel(\DL)))$ of proximity lattices with left-adjoint relations is dually equivalent to the category $\KOHaus$ of compact ordered spaces with continuous monotone maps.
\end{theorem}

\begin{remark}\label{rem:kohaus-prox-quasi-inverse}
  Under the equivalence, a compact ordered space $X$ corresponds to its proximity lattice $A_X$ (Definition~\ref{def:kohaus-proximity-lattice}) and $f:X\to Y$ to the rounded companion ${\prec_Y}\mathrel{;}(f^*)_\diamond\mathrel{;}{\prec_X}$ of the proximity morphism $f^*$ (Proposition~\ref{prop:kohaus-to-prox}).
  In particular, every proximity lattice is relationally isomorphic to the strong proximity lattice $A_X$ of its dual compact ordered space $X$.
\end{remark}

To conclude this section, we give a purely algebraic proof that every proximity lattice is relationally isomorphic (but not functionally isomorphic) to a strong one. 
The construction is reminiscent of the way \cite{JungSunderhauf1996} use pairs $(O,K)$ in order to get a strong proximity lattice from a stably compact space, see also Remark~\ref{rem:jung-sunderhauf-OK}. 

\begin{thm}\label{thm:strongification}
  In $\Map(\Spl_\le(\dRel(\DL)))$, every proximity lattice is naturally isomorphic to a strong proximity lattice.
\end{thm}
\begin{proof}
  For a proximity lattice $(A,\prec_A)$, 
  define $A^*$ to consist of all pairs $(I,F)$ where $I$ is a round ideal, $F$ is a round filter, and $a\in I$, $b\in F$ implies $a\prec_A b$. 
  $A^*$ is a distributive lattice ordered by inclusion in the ideals and reverse inclusion in the filters. More precisely, if $F$ and $G$ are round filters, then $F\sqcup G$ -- the filter generated by $F\cup G$ -- is round, as is $F\cap G$.
  Likewise, (ideal) joins and intersections of round ideals are round. Thus the lattice operations in $A^*$ are
  \begin{align*}
      (I,F)\wedge (J,G) &= (I\cap J,F\sqcup G)\\
      (I,F)\vee(J,G) &= (I\sqcup J,F\cap G)\\
      \bot_{A^*} &= (\{\bot_A\},A)\\
      \top_{A^*} &= (A,\{\top_A\})
  \end{align*}
Define $\prec_{A^*}$ by $(I,F)\prec_{A^*}(J,G)$ if and only if  $F\cap J\neq\emptyset$.

Evidently, $\prec_{A^*}$ is a distributive lattice weakening relation, and if $(I,F)\prec_{A^*}(J,G)$ then $I\subseteq J$ and $F\supseteq G$. 

To show that $(A^*,\prec_{A^*})$ is join strong, suppose $$(I,F)\prec_{A^*} (J,G)\vee (K,H).$$
Find elements $b\in J$ and $c\in K$ such that $b\vee c \in F$. 
By roundness of $J$ and $K$, find $b',b''\in J$ and $c',c''\in K$ such that $b\prec_A b'\prec_Ab''$ and $c\prec_A c'\prec_A c''$. 
Thus $b''$ witnesses that $(\dda b',\uua b') \prec_{A^*} (J,G)$ and $c''$ witnesses that $(\dda c',\uua c')\prec_{A^*} (K,H)$. 
Therefore $b\vee c \in F \cap (\dda b' \sqcup \dda c')$,
which is exactly $(I,F) \prec_{A^*} (\dda b',\uua b')\vee (\dda c',\uua c')$.
This is $\vee$-strongness: interpolants below $(J,G)$ and $(K,H)$ with $(I,F)$ below their join. Dually, $\prec_{A^*}$ is also meet strong.
     
Now define distributive lattice weakening relations ${\Vdash}\colon A\looparrowright A^*$ and ${\Vdash'}\colon A^*\looparrowright A$ by suitable restrictions of membership: $a\Vdash (I,F)$ if and only if $a\in I$ and $(I,F)\Vdash' a$ if and only if $F\ni a$. 
Roundness of $I$ and $F$ ensures that these are compatible with $\prec_A$ and $\prec_{A^*}$. The relations are adjoint, ${\Vdash}\dashv{\Vdash'}$: 
(Unit) If $a\prec_A b$, by interpolation find $c$ in between. Then $a\Vdash (\dda c,\uua c)\Vdash' b$.
(Counit) If $(I,F)\Vdash' c\Vdash (J,G)$, then by definition $c\in F\cap J$, so $(I,F)\prec_{A^*} (J,G)$.
The adjunction is a relational isomorphism: if $a\Vdash (I,F)\Vdash' b$, then $a\in I$ and $b\in F$ implies $a\prec_A b$, so ${\prec_A}={\Vdash}\mathrel{;}{\Vdash'}$; and if $(I,F)\prec_{A^*} (J,G)$, then for $c\in F\cap J$ we have $(I,F)\Vdash' c\Vdash (J,G)$, so ${\prec_{A^*}}={\Vdash'}\mathrel{;}{\Vdash}$.
\end{proof}
\begin{corollary}\label{cor:round-filters-strongification}
The lattice of round filters (dually, round ideals) of a proximity lattice $A$ is isomorphic to the lattice of round filters (dually, round ideals) of $A^*$.
In particular, the dual compact ordered spaces of $A$ and $A^*$ are isomorphic.
\end{corollary}
\begin{proof}
Let ${\Vdash}\colon A\looparrowright A^*$ and ${\Vdash'}\colon A^*\looparrowright A$ be the relational isomorphism of Theorem~\ref{thm:strongification}.
For a round filter $F$ in $A$ define
\[
  \Vdash(F) \;=\; \bigl\{(I,G)\in A^*\ \mid\ \exists a\in F.\, a\Vdash(I,G)\bigr\},
\]
and for a round filter $\mathcal{F}$ in $A^*$ define
\[
  \Vdash'(\mathcal{F}) \;=\; \bigl\{a\in A\ \mid\ \exists (I,G)\in \mathcal{F}.\, (I,G)\Vdash' a\bigr\}.
\]
These maps send round filters to round filters, and their compositions are identities because ${\Vdash}\dashv{\Vdash'}$ is a relational isomorphism.
The dual claim for round ideals is symmetric.
\end{proof}

\section{Conclusion}\label{sec:conclusion}

The work reported here is closely related to the duality-theoretic perspective on canonical extension developed by Mai Gehrke and collaborators. 
Priestley duality supplies the zero-dimensional dual $X$ of a distributive lattice $A$; the canonical extension is then $A^\delta=\mathrm{Up}(X)$, the complete lattice of all up-sets of the underlying poset of $X$. 
In that step one retains the dual order and discards the Priestley topology.

Our constructions likewise start from Priestley duality and dualise relations.
The double category of relations lifts $\Pri\simeq\DL^\op$ to $\dRel(\Pri)\simeq\dRel(\DL)^\op$.
Splitting idempotent proximity relations (below identity on the lattice side) freely adds objects dual to quotients by closed equivalence relations on the dual spaces.
Proximity relations and Priestley preorders are thus relational structure on the same dual data that underlie canonical extension.

Because our double-categorical treatment generalises to a large class of order-enriched natural dualities, we expect that it will be fruitful to investigate the interaction of this work with the theory of canonical extensions of distributive lattices, posets, and related structures.

\section*{Acknowledgements}

We thank the anonymous reviewer for careful reading and for comments that improved the paper.

\bibliographystyle{alpha}
\bibliography{main}

\end{document}